\documentclass[letterpaper, 11pt]{article}
\usepackage{graphicx} 
\usepackage[margin=2.15cm]{geometry}
\usepackage{algorithm}
\usepackage{algpseudocode}
\usepackage{amsthm}
\usepackage{amssymb,amsmath}  
\usepackage{thmtools}  
\usepackage{mathtools}
\usepackage[colorlinks=true, allcolors=blue]{hyperref}
\usepackage[capitalize]{cleveref}
\usepackage[utf8]{inputenc}
\usepackage[T1]{fontenc}
\usepackage{lmodern}

\declaretheoremstyle[headfont=\itshape,notefont=\bfseries,bodyfont=\normalfont]{claimproofstyle}
\declaretheorem[style=claimproofstyle,numbered=no,name=Proof of claim]{claimproof}

\usepackage{todonotes}
\usepackage{placeins}
\usepackage{subcaption}
\usepackage{booktabs}
\usepackage{authblk}
\usepackage{xspace}
\usepackage{newtxtext}
\usepackage{newtxmath}
\usepackage{tikz}
\usetikzlibrary{calc, shapes, backgrounds, graphs, arrows.meta, bending}
\usepackage{enumitem}
\setlist[1]{labelindent=\parindent}
\setlist[enumerate]{label=(\arabic*)}
\newtheorem{example}{Example}
\newtheorem{definition}{Definition}[section]
\newtheorem{lemma}{Lemma}[section]
\newtheorem{theorem}{Theorem}[section]
\newtheorem{claim}{Claim}
\newtheorem{fact}{Fact}
\newtheorem{corollary}{Corollary}[section]
\newtheorem{prop}{Proposition}
\newtheorem{observation}{Observation}[section]

\newcommand{\SkylineProblem}{\textsc{Minimum Skyline Point}\xspace}
\newcommand{\TwoTreesProblem}{\textsc{Two Trees SUS}\xspace}

\newcommand{\PairSUS}{\textsc{Shortest Unique String Pair}\xspace}
\newcommand{\SUS}{\textsc{Shortest Unique Substring}\xspace}
\newcommand{\SAS}{\textsc{Shortest Absent Substring}\xspace}
\newcommand{\SUSDiff}{\textsc{Shortest Exclusive Substring}\xspace}
\newcommand{\SkylineDiff}{\textsc{Minimum Exclusive Point}\xspace}
\newcommand{\TwoFamilies}{\textsc{Two String Families LCP}\xspace}
\newcommand{\LCP}{\textsf{LCP}}
\newcommand{\PairSES}{\textsc{Shortest Exclusive String Pair}\xspace}

\def\dd{\mathinner{.\,.}}
\newcommand{\cO}{\mathcal{O}}
\newcommand{\proofsubparagraph}[1]{\paragraph{#1}}

\newcommand{\Occ}{\mathrm{Occ}}
\newcommand{\per}{\textsf{per}}
\newcommand{\sd}{\textsf{sd}}
\newcommand{\str}{\textsf{str}}
\newcommand{\ST}{\textsf{ST}}
\newcommand{\TRIE}{\textsf{TRIE}}

\newcommand{\LCPs}{\mathsf{LCPs}}
\renewcommand{\P}{\mathcal{P}}
\newcommand{\Q}{\mathcal{Q}}
\renewcommand{\L}{\mathcal{L}}
\newcommand{\R}{\mathcal{R}}

\newcommand{\Lroot}{\textsf{Lroot}}
\newcommand{\Lrep}{\textsf{Lrepr}}

\newcommand{\Sync}{\mathbf{Sync}\xspace}
\newcommand{\maxLCP}{\mathsf{maxLCP}}

\newcommand{\prev}{\mathsf{prev}}
\newcommand{\nextt}{\mathsf{next}}

\makeatletter
\renewcommand{\paragraph}{%
  \@startsection{paragraph}{4}%
  {\z@}{0.5ex \@plus 1ex \@minus .2ex}{-1em}%
  {\normalfont\normalsize\bfseries}%
}
\makeatother

\newcommand{\defproblem}[3]{
\vspace{2mm}
\noindent\fbox{
   \begin{minipage}{0.96\textwidth}
   \textsc{#1}\\
   {\bf{Input:}} #2  \\
   {\bf{Output:}} #3
   \end{minipage}
   }
   \vspace{2mm}
}

\newcommand{\DrawRun}[4]{
	\def\h{0.5};
	\def\ext{0.4};

    \filldraw [fill=#4] (#1, #2)
		-- +(-\ext, 0)
		-- +(-\ext - 0.1, 0.5 * \h)
        -- +(-\ext, \h)
		-- +(#3 + \ext, \h)
		-- +(#3 + \ext + 0.1, 0.5 * \h)
        -- +(#3 + \ext, 0)
		-- +(0, 0);

    \foreach \i in {0,...,#3} {
        \draw (#1 + \i, #2) -- +(0, \h);
    }
}

\title{Faster Algorithms for Shortest Unique or Absent Substrings}

\author[1]{Panagiotis Charalampopoulos}
\author[1]{Manal Mohamed}
\author[2,3,4]{Solon~P.~Pissis}
\author[3]{Hilde Verbeek}
\author[5]{\\Wiktor Zuba}

\affil[1]{King's College London, UK\\
    \texttt{$\{$p.charalampopoulos, manal.1.mohamed$\}$@kcl.ac.uk}}
\affil[2]{The Cyprus Institute, Nicosia, Cyprus\\
    \texttt{s.pissis@cyi.ac.cy}}
\affil[3]{CWI, Amsterdam, The Netherlands\\
    \texttt{hilde.verbeek@cwi.nl}}
\affil[4]{Vrije Universiteit, Amsterdam, The Netherlands}
\affil[5]{University of Warsaw, Poland\\
    \texttt{w.zuba@mimuw.edu.pl}}
\date{\vspace{-5ex}}

\begin{document}
\maketitle

\begin{abstract}
We revisit two well-known algorithmic problems on strings:
computing a \emph{shortest unique substring} (SUS) and a \emph{shortest absent substring} (SAS) of a string $S$ of length $n$. Both problems admit folklore $\cO(n)$-time solutions using the suffix tree of $S$. However, for small alphabets, this complexity is not necessarily optimal in the word RAM model, where a string of length $n$ over alphabet $[0,\sigma)$ can be stored in $\cO(n \log \sigma/\log n)$ space and read in $\cO(n \log \sigma/\log n)$ time.

We present an $\cO(n \log \sigma/\sqrt{\log n})$-time algorithm for computing a SUS of $S$.
This algorithm decomposes the problem according to the length and the period of the sought substring and uses several tools and techniques, such as synchronizing sets, the analysis of runs, and wavelet trees, to reduce the computation of a SUS to a simple geometric problem.
Further, we adapt this algorithm and combine it with an efficient construction of de Bruijn sequences in order to obtain an $\cO(n \log \sigma/\sqrt{\log n})$-time algorithm for computing a SAS of $S$. 
\end{abstract}

\section{Introduction}\label{sec:intro}

Given a string $S$ over an alphabet $\Sigma$, a string $P$ is called
\emph{unique} in~$S$ if it occurs exactly once in~$S$ as a substring.
A unique substring $P$ of $S$ is called a \emph{shortest unique substring} (SUS) of~$S$ if there is no shorter string that is unique in $S$.
Similarly, a string $P$ over $\Sigma$ is said to be \emph{absent} from~$S$ if it does not occur in $S$ as a substring.
A string~$P$ that is absent from~$S$ is called a \emph{shortest absent substring} (SAS) of $S$ if no shorter string that is absent from $S$ exists.
For example, for string $S=\texttt{gcattgcgtaggt}$ over the alphabet $\Sigma=\{\texttt{a},\texttt{c},\texttt{g},\texttt{t}\}$, we have that $\texttt{ag}$ is a SUS of $S$ and $\texttt{aa}$ is a SAS of $S$.

Shortest unique substrings (SUSs) and shortest absent substrings (SASs) have a wide range of applications in bioinformatics, information retrieval, and data compression. In bioinformatics, SUSs are employed in alignment-free sequence comparison methods~\cite{DBLP:journals/bmcbi/HauboldPMW05}, whereas in information retrieval they are used to extract minimal text snippets from a document collection containing a query term~\cite{DBLP:conf/icde/PeiWY13}.
Note that the notion of SUS used in~\cite{DBLP:conf/icde/PeiWY13} is position-dependent: for a given position in the input string, one seeks a SUS covering that position. In contrast, in this work we consider the global variant of the problem, where the goal is to compute a shortest substring that is unique in the entire string.
SASs are likewise of significant importance in bioinformatics and data compression. They provide highly specific genomic signatures for pathogens such as SARS-CoV-2~\cite{DBLP:journals/bioinformatics/PratasS21}, thereby facilitating the development of rapid diagnostic assays and targeted therapeutics~\cite{DBLP:journals/bioinformatics/SilvaPCPF15}. Moreover, SASs are utilized in alignment-free sequence comparison methods~\cite{DBLP:journals/iandc/Charalampopoulos18} and constitute the foundational concept of compression schemes based on  antidictionaries~\cite{DBLP:journals/pieee/CrochemoreMRS00}.

It is a classical exercise to compute a SUS or a SAS of a string $S$ of length $n$ in $\cO(n)$ time using the suffix tree of $S$~\cite{DBLP:conf/focs/Weiner73}.
In particular, all SUSs and SASs are naturally encoded  in the suffix tree of $S$; see \cref{fig:SUS-SAS} for an illustration.

\begin{figure}[ht]
    \centering
    \includegraphics[width=0.75\linewidth]{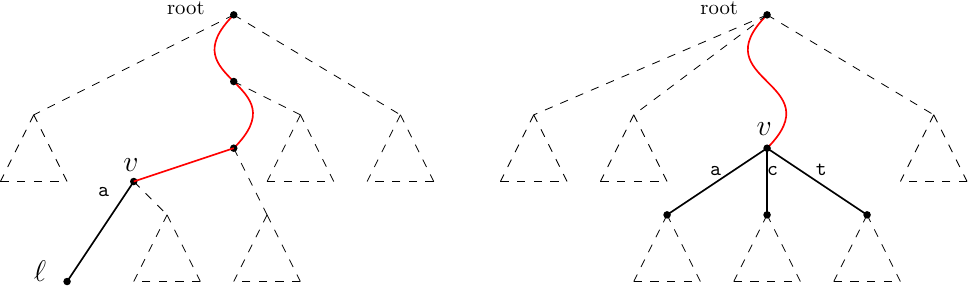}
    \caption{A schematic illustration of two suffix trees.
    Left, the root-to-$v$ path-label, appended with letter $\texttt{a}$, is a \emph{unique substring} that occurs as a prefix of the suffix represented by the leaf $\ell$.
    Right, the root-to-$v$ path-label, appended with letter $\texttt{g}\notin\{\texttt{a},\texttt{c},\texttt{t}\}$, yields an \emph{absent substring}.}
    \label{fig:SUS-SAS}
\end{figure}

Here, we consider the unit-cost word RAM model of computation with word size $w= \Theta(\log n)$ for inputs of size $n$, and a standard instruction set including arithmetic operations, bitwise Boolean operations, and shifts.
We measure the space complexity of our algorithms in terms of the number of machine words used.
A \emph{packed representation} of a string $S$ over an integer alphabet $\Sigma=[0,\sigma)$ stores $\Theta(\log_\sigma n)$ letters per machine word (possibly apart from the last one), thus representing $S$ in $\cO(1+|S|/\log_{\sigma}n)$ machine words.
A string given in this representation is referred to as a \emph{packed string}.

A large body of work exploits bit-level parallelism in the word RAM model to accelerate classical string processing tasks when the input consists of packed strings.
In particular, the problems for which speed-ups have been obtained include
pattern matching~\cite{DBLP:journals/spe/Baeza-Yates89,DBLP:conf/iwoca/Belazzougui10,DBLP:journals/tcs/Ben-KikiBBGGW14,DBLP:journals/jda/Bille11,DBLP:conf/cpm/BreslauerGG12,DBLP:journals/ipl/Fredriksson03,DBLP:journals/ipl/GrabowskiF08,DBLP:conf/cpm/NavarroR98},
text indexing~\cite{DBLP:conf/cpm/BilleGS17,DBLP:conf/cpm/0001G15,DBLP:journals/ieicet/TakagiISA17}, computing palindromes~\cite{DBLP:conf/cpm/Charalampopoulos22}, constructing longest common extension data structures~\cite{DBLP:conf/stoc/KempaK19}, computing a longest common substring~\cite{TALG-Charalampopoulos25},
constructing the BWT~\cite{DBLP:conf/stoc/KempaK19},
computing Lempel-Ziv (LZ77) factorizations~\cite{DBLP:conf/spire/Ellert23,DBLP:conf/focs/KempaK24},
counting squares and locating runs~\cite{DBLP:conf/mfcs/Charalampopoulos25}, computing Lyndon arrays~\cite{DBLP:conf/esa/BannaiE23},
computing covers~\cite{DBLP:conf/spire/RadoszewskiZ24}, and constructing compressed suffix trees and compressed suffix arrays~\cite{DBLP:conf/soda/KempaK23}.

In this paper, we ask the following question for a string of length $n$: 
\begin{center}
    \emph{Can a SUS (or a SAS) be computed in $o(n)$ time in the packed setting?}
\end{center}

\subparagraph{Our results.}
We answer this question in the affirmative, improving upon the folklore linear-time solutions for computing a single SUS or SAS in the packed setting.

\begin{restatable}[SUS]{theorem}{sus}\label{the:sus}
Given a packed string $S$ of length $n$ over an integer alphabet $[0,\sigma)$, 
a~shortest unique substring of $S$ can be computed in $\cO({n\log \sigma}/{\sqrt{\log n}})$ time.
\end{restatable}
\begin{restatable}[SAS]{theorem}{sas}
\label{the:sas}
Given a packed string $S$ of length $n$ over an integer alphabet $[0,\sigma)$, 
a~shortest absent substring of $S$ can be computed in $\cO({n\log \sigma}/{\sqrt{\log n}})$ time.
\end{restatable}

\subparagraph{Our techniques and paper organization.} 
We obtain \cref{the:sus} by decomposing the problem based on the \emph{length} and the \emph{period} of the sought substring of $S$. 
For short substrings, we employ tabulation. 
For medium-length aperiodic substrings, we modify an existing technique for longest common substrings~\cite{TALG-Charalampopoulos25} that is based on an efficient construction of wavelet trees~\cite{DBLP:journals/tcs/MunroNV16}.
For long aperiodic substrings, we combine sampling via string synchronizing sets~\cite{DBLP:conf/stacs/EllertK26,DBLP:conf/stoc/KempaK19} with a heavy-light decomposition~\cite{DBLP:journals/jcss/SleatorT83} of two compacted tries $T$ and $T^R$. Then, using pairs of heavy paths, we reduce the search to a 2D geometric problem that, as we show, underlies the computation of SUSs.
For periodic substrings, we exploit the runs in $S$~\cite{DBLP:journals/tcs/IliopoulosMS97}: we group them by their Lyndon roots and sparse-Lyndon roots~\cite{DBLP:conf/mfcs/Charalampopoulos25} and use this grouping to reduce the search to another instance of said geometric problem.
See \Cref{sec:SUS}. We also show that any length-$n$ SUS instance over $[0,\sigma)$ can be reduced in $\cO(\frac{n}{\log_{\sigma} n})$ time to an instance of length $\cO(n\log\sigma)$ over a binary alphabet.
This implies that the binary alphabet constitutes a hardest case for computing a SUS. In particular, any algorithm designed for the binary case can be applied to general alphabets via this reduction,
without any increase in the asymptotic running time. See \Cref{sec:bin}.

We obtain \cref{the:sas} using the same framework, 
augmented with a new efficient construction of de Bruijn sequences~\cite{deBruijn1946} in packed representation. See \Cref{sec:SAS}.

We start in \Cref{sec:prel} with the necessary preliminaries. 
Apart from the preliminaries, we also sketch
the folklore linear-time solutions for computing all SUSs and SASs using suffix trees.

\subparagraph{Other related work.}
Kempa and Kociumaka~\cite{DBLP:conf/stoc/KempaK25} studied the hardness hierarchy for problems whose fastest known word-RAM algorithms run in  $\cO(n\sqrt{\log n})$ time on inputs of $\Theta(n)$ machine words.
This class includes  string processing problems such as constructing the BWT~\cite{DBLP:conf/stoc/KempaK19}, computing a longest common substring~\cite{TALG-Charalampopoulos25}, and computing LZ77 factorizations~\cite{DBLP:conf/focs/KempaK24}.
For these problems, they showed that the known $\cO(n/\sqrt{\log n})$-time algorithms for \emph{binary} strings---equivalently, $\cO(n\sqrt{\log n})$-time algorithms when the input size is measured in machine words---are conditionally optimal, with conditional lower bounds established via non-trivial reductions to a variant of dictionary matching.
Computing a SUS (\cref{the:sus}) or a SAS (\cref{the:sas}) for binary strings shares the same time complexity.
It remains an interesting open problem whether one can improve upon our algorithms or prove that they are conditionally optimal.
Let us note that proving a conditional lower bound for our problem using the framework of Kempa and Kociumaka~\cite{DBLP:conf/stoc/KempaK25} seems challenging. The reductions in said work highlight as the underlying hardness in the studied problems the task of looking for an occurrence of some (sub)string $P$ in a string $S$ (e.g., as in the longest common substring problem). The problems we study here are of a different flavor: in SUS we look for substrings that do not occur elsewhere in $S$ while in SAS we look for strings that do not occur at all in $S$.

\section{Preliminaries}\label{sec:prel}

\subparagraph{Strings.} An \emph{alphabet} $\Sigma$ is a finite set of elements called \emph{letters}. 
A \emph{string} $S=S[0] S[1] \cdots S[n-1]$ of \emph{length} $|S|=n$ is a sequence of $n$ letters from $\Sigma$.
We refer to each $i\in [0,n)$ as a \emph{position} of $S$.
We consider throughout an \emph{integer} alphabet $\Sigma=[0,\sigma)$ with $\sigma=n^{\cO(1)}$. 
A string $P$ is a \emph{substring} of a string $S$ if we have $P=S[i] \cdots S[i+|P|-1]$ for some position $i$ of $S$.
In this case, we say that $P$ \emph{occurs} at position $i$ of $S$.
Such an occurrence is called a \emph{fragment} of~$S$; we denote it by either $S[i\dd i+|P|)$ or $S[i\dd i+|P|-1]$.
The set of the starting positions of the \emph{occurrences} of a string $P$ in $S$ is denoted by $\Occ(P,S)$.
A \emph{prefix} of $S$ is a fragment of the form $S[0\dd j)$, and a  \emph{suffix} of $S$ is a fragment of the form $S[i \dd n)$.
A substring $P$ of~$S$ is called \emph{proper} if $P \neq S$.
The \emph{reverse} of string $S$, which we denote by $S^R$, is defined as $S^R=S[n-1]S[n-2]\cdots S[0]$.
 The \emph{concatenation} of two strings $S$ and  $S'$ is denoted by $S\cdot S'$ and the concatenation of $k$ copies of a string $S$ is denoted by $S^k$.

\begin{definition}[Period]
An integer $p>0$ is a \emph{period} of a string $P$ if $P[i] = P[i + p]$, for all $i \in [0,|P|-p)$. The smallest period of $P$ is called \emph{the period} of $P$ and is denoted by $\per(P)$. A string $P$ is called \emph{periodic} if $\per(P)\leq |P|/2$ and \emph{aperiodic} otherwise.\label{def:periodic}
\end{definition}

\begin{example}
For string $P=\texttt{abaaabaaabaaaba}$, we have $\per(P)=4$. Note that $8$ is also a period of $P$. Since $\textsf{per}(P)=4\leq |P|/2=15/2$, $P$ is periodic.   
\end{example}

\begin{definition}[Run]\label{def:run}
A fragment $F=S[i\dd j]$ of a string $S$ is called a \emph{run} of $S$ if 
$\per(F) \leq |F|/2$, and extending $F$ either to the left or to the right (if possible) would result in an increase of its period, that is,
$S[i-1] \ne S[i-1+\per(F)]$ (or $i = 0$) and $S[j+1] \ne S[j+1-\per(F)]$ (or $j=|S|-1$). 
\end{definition}
\begin{definition}[Lyndon Root]
The \emph{Lyndon root} of a run $F$, denoted by $\Lroot(F)$, is defined as the lexicographically smallest rotation of the prefix $F[0 \dd \textsf{per}(F))$ of $F$.
\end{definition}

\begin{definition}[Lyndon Representation]\label{def:canonical} The Lyndon representation of a run $F$ is a quadruple $\Lrep(F) =
(\lambda, e, \alpha, \beta)$ such that:
\begin{itemize}
\item $\lambda=\Lroot(F)$, and
\item $F=P\cdot \lambda^e \cdot T$, where $P$ is a (possibly empty) suffix of $\lambda$ with $|P|=\alpha< |\lambda|$, and $T$ is a (possibly empty) prefix of $\lambda$ with $|T|=\beta< |\lambda|$.
\end{itemize}
\end{definition}

\begin{example}
Let $S=\texttt{\underline{abaaabaaabaaaba}b}$. The \emph{underlined} fragment $F = S[0\dd 14]=\texttt{abaaabaaabaaaba}$ is a run with $\per(F)=4\leq |F|/2=15/2$. Observe that extending $F$ to the right yields the fragment $S[0\dd 15]=S$, whose period is $14$; hence extending $F$ increases its period.
Moreover, we have $\Lroot(F)=\texttt{aaab}$
and $\Lrep(F) = (\texttt{aaab},3, 2, 1)$.
\end{example}

\begin{definition}[$\tau$-Run]
    A run $R$ of a string $S$ is called a \emph{$\tau$-run}, for an integer $\tau >0$,
    if $|R| \ge 3\tau - 1$ and $\per(R) \le \frac13 \tau$.
\end{definition}

\begin{lemma}[Lemma 2.8~\cite{TALG-Charalampopoulos25}]
\label{lem:per:find-roots}
Let $S$ be a string of length $n$ over an integer alphabet $[0,\sigma)$, with $\sigma=n^{\cO(1)}$, and let  $\tau>0$ be an integer. Then $S$ contains 
$\cO(n/\tau)$ $\tau$-runs.
Moreover, if $\tau \leq \frac{1}{4} \log_{\sigma} n$, all $\tau$-runs in
$S$ can be computed and grouped by their Lyndon roots in $\cO(n/\tau)$ time. 
Within the same time bound, for each $\tau$-run, we can compute the two leftmost occurrences of its Lyndon root.
\end{lemma}

\begin{definition}[Shortest Unique Substring (SUS)]
A string $P$ is a \emph{unique substring} of a string $S$ if and only if $|\Occ(P,S)|=1$.
A unique substring $P$ of $S$ is called a \emph{shortest unique substring} of $S$
if there is no string $P'$ with $|P'|<|P|$ that is a unique substring of $S$.   
\end{definition}

\begin{definition}[Shortest Absent Substring (SAS)]
A string $P$ is \emph{absent} from a string $S$ if and only if $|\Occ(P,S)|=0$.
An absent string $P$ of $S$ is called a \emph{shortest absent substring} of~$S$,
if there is no string $P'$ with $|P'|<|P|$ that is absent from $S$.
\end{definition}

\subparagraph{String synchronizing sets.} We next define a powerful sampling mechanism.
\begin{definition}[String Synchronizing Set~\cite{DBLP:conf/stoc/KempaK19}]\label{def:SSS}
For a string $S$ of length $n$ and a positive integer $\tau  \in  [1,\lfloor  \frac{n}{2} \rfloor]$, 
a set $\Sync \subseteq  [0, n-2\tau]$ is a \emph{$\tau$-synchronizing set} of $S$ if it satisfies:
\begin{enumerate}
    \item Consistency: For $i,j \in [0, n-2\tau]$, if $S[i \dd i + 2\tau) = S[j \dd j + 2\tau)$, then  $i \in  \Sync$  if and only if $j \in  \Sync$.
    \item Density: For $i \in  [0 , n-3\tau +1]$, $\Sync \cap [i,i + \tau ) = \emptyset$
    if and only if $\per(S[i\dd i + 3\tau -1)) \leq  \frac{1}{3}\tau$.
\end{enumerate}
\end{definition}

\begin{theorem}[\cite{DBLP:conf/stacs/EllertK26}]\label{the:sss}
A string $S \in [0, \sigma)^n$ can be preprocessed in $\cO(\frac{n}{\log_{\sigma} n})$ time so that, given 
$\tau \in [1, \lfloor n/2 \rfloor]$, a $\tau$-synchronizing set $\Sync$ of $S$ of size 
$|\Sync| < \tfrac{70n}{\tau}$ can be constructed in $\cO(n/\tau)$ time.
\end{theorem}

\subparagraph{Problem definitions.}
We now formally  define the problems in scope. 

\defproblem{\SUS}
{A packed string $S$ of length $n$ over an integer alphabet $\Sigma=[0,\sigma)$.}
{$i,j \in \mathbb{Z}_{\geq 0}$ such that $P=S[i\dd j]$ is a shortest unique substring of $S$.}

\defproblem{\SAS}
{A packed string $S$ of length $n$ over an integer alphabet $\Sigma=[0,\sigma)$.}
{$i,j \in \mathbb{Z}_{\geq 0}$ and $c\in\Sigma$ such that $P=S[i\dd j]\cdot c$ is a shortest absent substring of~$S$.}

\subparagraph{Compacted tries and suffix trees.}
For a set $\mathcal{S}$ of strings over $\Sigma$, the \emph{trie} $\TRIE(\mathcal{S})$ is a rooted tree whose nodes are in one-to-one correspondence with the set of the prefixes of the strings in $\mathcal{S}$.
Each edge of $\TRIE(\mathcal{S})$ is labeled with a letter from $\Sigma$.
The string $\str(v)$ represented by a node $v$ is the concatenation of the labels on the edges along the root-to-$v$ path; the node $v$ is called the \emph{locus} of $\str(v)$. 
The order on $\Sigma$ induces a corresponding order on the outgoing edges  of every node in $\TRIE(\mathcal{S})$.
A node $v$ is \emph{branching} if it has at least two children and \emph{terminal} if $\str(v) \in \mathcal{S}$. 

A \emph{compacted trie} for a set $\mathcal{S}$ of strings is obtained from $\TRIE(\mathcal{S})$ by dissolving all non-root nodes except for the branching and the terminal nodes. The removed nodes are called \emph{implicit} and the preserved ones are called \emph{explicit}.
The edges of the compacted trie are labeled by fragments of elements of $\mathcal{S}$ rather than single letters.
The \emph{string depth} $\sd(v)$ of a node $v$ is defined as $\sd(v)=|\str(v)|$,  i.e., it is the length of the string represented by $v$, or, equivalently, the total length of the labels along the root-to-$v$ path.
The compacted trie requires $\cO(|\mathcal{S}|)$ space if we have random access to elements of~$\mathcal{S}$, as any fragment $S[i\dd j]$ for $S \in \mathcal{S}$ can be stored in $\cO(1)$ space using $i$,~$j$, and a handle to $S$.

The \emph{suffix tree} $\ST(S)$ of a string $S$ is the compacted trie for the set of suffixes of $S$. Each terminal node $v$ of $\ST(S)$ is labeled by $|S|-\sd(v)$, i.e., the starting position of the suffix it represents; see \cref{fig:st} for an example. The suffix tree $\ST(S)$ can be constructed in $\cO(n)$ time for any string of length $n$ over an integer alphabet $\Sigma=[0, \sigma)$ with $\sigma=n^{\cO(1)}$~\cite{DBLP:conf/focs/Farach97}.

\subparagraph{Wavelet trees.}
For an arbitrary alphabet $\Sigma$, a \emph{skeleton tree} for $\Sigma$
is a full binary tree $\mathcal{T}$ together with a bijection between $\Sigma$ and the leaves of $\mathcal{T}$. 
For a node $v\in \mathcal{T}$, $\Sigma_v$ denotes the subset of $\Sigma$ that corresponds to the leaves in the subtree of $v$.

Given a skeleton tree $\mathcal{T}$ and a string $S\in \Sigma^*$, the \emph{$\mathcal{T}$-shaped wavelet tree} of $S$ is the tree~$\mathcal{T}$ augmented with, for each internal node $v$, a bit-vector $B_v$ described below.
For each node $v$ of said tree, let $S_v$ denote the subsequence of $S$ that consists of letters from $\Sigma_v$.
Further, denote the left and right children of each internal node $v$ by $v_L$ and~$v_R$, respectively.
The bit-vector $B_v$ is of size $|S_v|$ and $B_v[i]=1$ if and only if $S_v[i]\in \Sigma_{v_R}$; see \cref{fig:wt} for an example.
Wavelet trees were introduced in~\cite{DBLP:conf/soda/GrossiGV03}, where an $\cO(n \log \sigma)$-time construction algorithm was presented. 
More efficient algorithms for constructing wavelet trees were later proposed in~\cite{DBLP:journals/tcs/MunroNV16,DBLP:conf/soda/BabenkoGKS15}.

\begin{figure}[t]
    \centering
    \begin{subfigure}{0.45\textwidth}
        \centering
        \centering
        \includegraphics[width=0.85\linewidth]{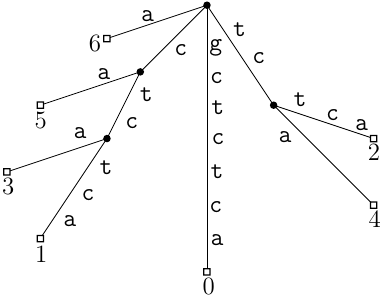}
        \caption{The suffix tree of string $S=\texttt{gctctca}$.}\label{fig:st}
    \end{subfigure}
    \begin{subfigure}{0.45\textwidth}
        \centering
        \includegraphics[width=1\linewidth]{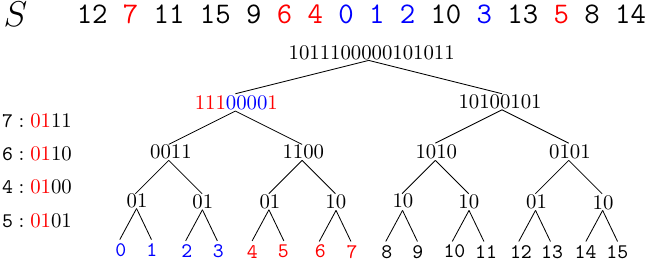}
        \caption{The $\mathcal{T}$-shaped wavelet tree of string $S=\texttt{12~7~11~15~9~6~4~0~1~2~10~3~13~5~8~14}$; the bijection between the leaves of $\mathcal{T}$ and $\Sigma$ is encoded in the leaves.
        Let $v$ be the left child of the root.
        We have $B_v = 11100001$ since $\texttt{7}, \texttt{6}, \texttt{4}, \texttt{5}$ (shown in red) correspond to descendants of $v_R$ (right child), while $\texttt{0}, \texttt{1}, \texttt{2}, \texttt{3}$ (shown in blue) correspond to descendants of $v_L$ (left child).}\label{fig:wt}
    \end{subfigure}
    \caption{An illustration of a suffix tree and a wavelet tree.}
    \label{fig:st_wt}
\end{figure}
 
\subparagraph{Sketches of linear-time solutions.}

Given the suffix tree $\ST(S)$ of a string $S$ of length $n$, both \SUS and \SAS can be solved in $\cO(n)$ time.

\SUS:
Consider a terminal node $v$ of $\ST(S)$.
If $v$ is not a leaf, then no string that is unique in $S$ occurs at position $n-\sd(v)$ of $S$.
We now consider the case when $v$ is a leaf and denote its parent by~$u$.
The shortest string that is unique in $S$ and occurs at position $n-\sd(v)$ is $S[n-\sd(v) \dd n-\sd(v) + \sd(u)]$ and is of length $\sd(u)+1$. By iterating over all leaves of $\ST(S)$, we may thus compute all SUSs of $S$ in time $\cO(n)$.

\SAS:
Consider a node $u$ of $\ST(S)$. If $u$ has $\sigma$ outgoing edges, then there is no string~$P$ that is absent from $S$ and has $\str(u)$ as its longest proper prefix. However, if $u$ has fewer than $\sigma$ outgoing edges, then, for each letter $c$ such that $u$ has no outgoing edge whose string label starts with $c$, the string $\str(u)\cdot c$ is absent from $S$.
By iterating over all nodes of $\ST(S)$, we can compute the nodes $u$ with fewer than $\sigma$ outgoing edges and minimum string depth, and then output all SASs in time $\cO(n+|\textsf{output}|)$.

In what follows, we assume that all considered strings are given in packed representation.
When it is clear from the context, we may omit the term \emph{packed}.

\section{Computing a Shortest Unique Substring}\label{sec:SUS}

We decompose the problem into four cases based on the length $\ell$ and the period $p$ of a SUS:

\begin{enumerate}
    \item \textbf{Short case}: $\ell \leq \frac15 \log_\sigma n$ (see \Cref{sec:short});
    
    \item \textbf{Medium aperiodic case}: $\ell \in (\frac15 \log_\sigma n, 2^{\sqrt{\log n}})$ and $p > \frac{1}{45}\log_\sigma n$ (see \Cref{sec:medium});
    
    \item \textbf{Long aperiodic case}: $\ell \geq \log^4 n$ and $p > \frac19 \log^4 n$ (see \cref{sec:long});
    
    \item \textbf{Periodic case}: $\ell \ge 3\tau$ and $p \le \frac13 \tau$, for $\tau = \lfloor\frac{1}{15}\log_\sigma n \rfloor$ or $\tau = \lfloor \frac13\log^4 n \rfloor$  (see \cref{sec:periodic}).
    Note that when $\tau=0$, this case is obsolete because the period of any string is positive.
\end{enumerate}

In \Cref{sec:skyline}, we introduce and solve a simple geometric problem that serves as a subroutine for computing a SUS in both the long aperiodic and the periodic cases.

\subsection{Short Case}\label{sec:short}

\begin{lemma}
\label{lem:short}
Given an instance of \SUS, in $\cO(n/\log_\sigma n)$ time,
we can either compute a SUS of $S$ of length at most $\ell := \lfloor\frac15 \log_\sigma n\rfloor$ if one exists, or conclude that no SUS of length at most $\ell$ exists.
\end{lemma}
\begin{proof}
    We use the standard trick of covering $S$ with short overlapping fragments that contain all length-$\ell$ fragments of $S$ such that the number of distinct strings is strongly sublinear in~$n$.
    Specifically, we consider the fragment $S_i := S[i \dd \min \{i+2\ell,n\})$ for each $i \in [0,n-\ell)$ that is equivalent to $0$ modulo~$\ell$.
    There are $\cO(n/\ell)$ such fragments in $S$, and every length-$\ell$ 
    fragment of $S$ is covered by at least one of them.

    Over an alphabet of size $\sigma$, collection $\mathcal{S} := \{S_i \mid i \in [0,n-\ell) \text{ and } i \equiv 0 \pmod{\ell} \}$ contains $\cO(\sigma^{2\ell}) = \cO(n^{2/5})$ \emph{distinct} strings. Since $2\ell = \cO(\log_\sigma n)$, each string in $\mathcal{S}$ fits in $\cO(1)$ machine words.
    We lexicographically sort the elements of $\mathcal{S}$ by treating them as integers using bucket sort, and then keep at most \emph{two} copies of each distinct string in $\mathcal{S}$.
    This takes $\cO(n^{2/5} + n/\ell) =$ $\cO(n/\ell)$~time.
    Let $F_1\leq F_2 \leq \dots \leq F_q$, for $q=\cO(n^{2/5})$, denote the obtained sorted list.
    For each integer $i\in [1,q]$, let $p_i$ be the left endpoint of an occurrence of $F_i$ in $S$.

    Next, we construct the string $S' = \#_0 F_1\#_1 F_2\#_2 \dots\#_{q-1}F_q\#_{q}$ of length $|S'|=\cO(n^{2/5}\ell)$, where the letters $\#_0,\ldots,\#_{q} \notin~\Sigma$ are pairwise distinct.
    We map each position $j$ of $S'$ such that $S'[j] \in \Sigma$ to its original position in $S$.
    Formally, any position $j$ of $S'$ such that $j \in (j_1,j_2)$, 
    where $S'[j_1] = \#_{i-1}$ and $S'[j_2]=\#_{i}$, is mapped to position $p_i + (j-(j_1+1))$.
    
    Our goal is now to find a shortest substring $P$ of $S'$ that does not contain any letter $\#_i$ and such that, if $P$ occurs at distinct positions $j$ and $j'$ in $S'$, then positions $j$ and $j'$ correspond to the same original position in $S$.
    To this end, we construct a compacted trie over the suffixes of the strings $F_i\#_i$.
    This can be achieved by building the suffix tree $\ST(S')$ in $\cO(n^{2/5}\ell)$ time~\cite{DBLP:conf/focs/Farach97} and postprocessing it in linear time using (batched) weighted ancestor queries (see~\cite{DBLP:journals/jea/Charalampopoulos20}).
    Each leaf corresponding to a suffix $F_i[j \dd |F_i|)\#_i$ is labeled with $p_i+j$.
    Then, for each leaf $u$ whose parent $v$ does not have string depth $\sd(u)-1$, we introduce an artificial explicit node with that string depth along edge $(u,v)$.
    
    Finally, we perform a  post-order traversal of $\ST(S')$ to mark all internal nodes $u$ all of whose descendant leaves are marked with the same integer, and compute the minimal value of $\sd(\textsf{parent}(u)) + 1$ over all such nodes.
    If this value is at most $\ell$, we conclude that the SUSs of $S$ are of length at most $\ell$, and can return any of them: each SUS corresponds to the path-label of the path from the root to the (possibly implicit) node with string depth  $\sd(\textsf{parent}(u)) + 1$ in the root-to-$u$ path for a marked node $u$.
    Otherwise, we conclude that all substrings of $S$ of length at most $\ell$ have at least two occurrences in $S$.
   Said traversal takes time linear in the size of the compacted trie.
   Overall, the algorithm thus runs in $\cO(n^{2/5}\ell + n/\ell)=\cO(n/\log_{\sigma}n)$ time.
   
   \proofsubparagraph{Correctness.} By construction, the postprocessed trie $\ST(S')$ represents all the substrings of~$S$ of length at most~$\ell$ (plus some of the substrings of length at most $2\ell$). By keeping at most two copies of each distinct string in $\mathcal S$, we ensure that the (non-)uniqueness of substrings is preserved between $S$ and $S'$. In particular, if a substring $P$ of length (at most) $\ell$ is unique in~$S$, then all the leaf descendants of the node representing $P$ in $\ST(S')$  are labeled with the \emph{same position} of its single occurrence in $S$. Therefore, the algorithm needs to find the lowest such node (explicit or implicit) whose leaf descendants all share the same position label. The minimality condition can only be satisfied by a direct child of a branching node $\textsf{parent}(u)$  representing a non-unique substring, where $u$ itself is an explicit node representing a unique substring.
   Finally, by excluding from the output the leaves (that represent unique substrings containing letters $\#_i$), and unique substrings whose string depth exceeds $\ell$ (since their other occurrence might simply not be represented in $S'$), we ensure that the reported substring is indeed a genuine SUS of $S$.
\end{proof}

\subsection{Medium Aperiodic Case}\label{sec:medium}

We look for a SUS with length $\ell \in (\frac15 \log_\sigma n, 2^{\sqrt{\log n}})$ and period $p > \frac{1}{45}\log_\sigma n$.
Our solution for this case builds on 
a method for computing a longest common substring in the packed setting~\cite{TALG-Charalampopoulos25}.
This method relies on the \TwoFamilies problem, originally introduced in \cite{charalampopoulos_et_al:LIPIcs.CPM.2018.23}.
Let us denote the length of the longest common prefix (LCP) of two strings $U$ and $V$
by $\LCP(U, V)$.

\defproblem{\TwoFamilies}{%
Compacted tries $\mathcal{T}(\mathcal{F}_1)$ and $\mathcal{T}(\mathcal{F}_2)$ of $\mathcal{F}_1, \mathcal{F}_2 \subseteq \Sigma^*$,
and two sets $\mathcal{P},\mathcal{Q}\subseteq \mathcal{F}_1 \times \mathcal{F}_2$ 
with $|\mathcal{P}|,|\mathcal{Q}|,|\mathcal{F}_1|, |\mathcal{F}_2| \le N$.}
{$\max\{\LCP(P_1,Q_1)+\LCP(P_2,Q_2) : (P_1,P_2)\in \mathcal{P},(Q_1,Q_2)\in \mathcal{Q}\}$.}

We next define an analogous problem for SUSs: 

\defproblem{\PairSUS}{%
Compacted tries $\mathcal{T}(\mathcal{F}_1)$ and $\mathcal{T}(\mathcal{F}_2)$  of $\mathcal{F}_1, \mathcal{F}_2 \subseteq \Sigma^*$,
and a (multi)set $\mathcal{S}\subseteq \mathcal{F}_1\times \mathcal{F}_2$ 
with {$|\mathcal{S}|, |\mathcal{F}_1|, |\mathcal{F}_2| \le~N$}.}{Integers $\ell_1,\ell_2\ge0$ and a pair $(P_1,P_2) \in \mathcal{S}$, such that $(\ell_1 + \ell_2)$ is minimized and for all  $(Q_1, Q_2) \in \mathcal{S}\setminus \{(P_1,P_2)\}$, either $\LCP(P_1,Q_1) < \ell_1$ or $\LCP(P_2, Q_2) < \ell_2$.}

We further recall a few definitions from~\cite{TALG-Charalampopoulos25} necessary for stating the needed result from~\cite{TALG-Charalampopoulos25}.

\begin{definition}[Prefix Family]
    A set $\mathcal{S} = \{(U_1,V_1),\dots,(U_N,V_N)\}$ of string pairs is a \emph{prefix family} if there exists a string $Y$ such that $U_i$ is a prefix of $Y$ for every $i = 1,\ldots,N$.
\end{definition}

\begin{definition}[$(\alpha,\beta)$-Family]
    A set $\mathcal{S} = \{(U_1,V_1),\dots,(U_N,V_N)\}$ of string pairs is an \emph{$(\alpha,\beta)$-family} if, for all $(U,V) \in \mathcal{S}$, $|U| \le \alpha$ and $|V| \le \beta$.
\end{definition}

For each of these types of families, we design a different algorithm that solves \PairSUS when $\mathcal{S}$ is of said type of family.

\subsubsection{Solution for Prefix Families}
We consider an instance of \PairSUS with $\mathcal{S}$ being a prefix family $\{(U_1,V_1),\ldots,(U_N,V_N)\}$. Since the first components of the elements of $\mathcal{S}$ are prefixes of some common string $Y$, we have that $\LCP(U_i,U_j) = \min\{|U_i|,|U_j|\}$, for any $i,j \in [1,N]$. We will use this fact to first show a formula for \PairSUS on a prefix family.

\begin{lemma}\label{lem:med:prefix-formula}
    Consider an instance of \PairSUS in which $\mathcal{S}$ is a prefix family. 
    For any integer $\ell_1 > 0$, let $I_{\ell_1} = \{ i \mid |U_i| \ge \ell_1 \}$. Further let 
    \[ f(\ell_1) = \min_{i \in I_{\ell_1}} \left( \max_{j \in I_{\ell_1} \setminus \{i\}} \LCP(V_i, V_j) \right), \]
    with $f(\ell_1) = 0$ if $|I_{\ell_1}| = 1$. The pair $(\ell_1, f(\ell_1)+1)$ minimizing the sum $\ell_1 + f(\ell_1) + 1$ is an optimal solution.
\end{lemma}

\begin{proof}
    We show that the formulation is feasible and that any optimal solution must have a cost at least as large as the one proposed.

    \proofsubparagraph{Feasibility.} 
    For a fixed $\ell_1$, let $i^* \in I_{\ell_1}$ be an index attaining the minimum in $f(\ell_1)$, and let $\ell_2 = f(\ell_1) + 1$. We choose the pair $(P_1, P_2) = (U_{i^*}, V_{i^*})$ as our witness. To satisfy \PairSUS, every other pair $(Q_1, Q_2) = (U_j, V_j) \in \mathcal{S} \setminus \{(U_{i^*}, V_{i^*})\}$ must be separated. There are two cases for index $j$:
    \begin{itemize}
        \item Case 1: $j \notin I_{\ell_1}$. In a prefix family, all $U$ strings are prefixes of some common string. If $|U_j| < \ell_1 \le |U_{i^*}|$, then $U_j$ is a proper prefix of $U_{i^*}$, which implies $\LCP(U_{i^*}, U_j) = |U_j| < \ell_1$. Thus, the first condition of \PairSUS is satisfied.
        \item Case 2: $j \in I_{\ell_1}$. By the definition of $i^*$, we have $\LCP(V_{i^*}, V_j) \le \max_{k \in I_{\ell_1} \setminus \{i^*\}} \LCP(V_{i^*}, V_k) = f(\ell_1)$. Since $\ell_2 = f(\ell_1) + 1$, it follows that $\LCP(V_{i^*}, V_j) < \ell_2$. Thus, the second condition of \PairSUS is satisfied.
    \end{itemize}
    In both cases, $(U_{i^*}, V_{i^*})$ is a valid witness for $(\ell_1, \ell_2)$.

    \proofsubparagraph{Optimality.} 
    Let $(\ell_1^*, \ell_2^*)$ be an optimal solution with witness $(U_k, V_k)$. For this witness to be valid, we must have $|U_k| \ge \ell_1^*$. Furthermore, for the solution to be feasible, every other pair $j \in I_{\ell_1^*} \setminus \{k\}$ must satisfy the second condition of \PairSUS (since they fail the first: $|U_j| \ge \ell_1^* \implies \LCP(U_k, U_j) \ge \ell_1^*$). 
    
    Therefore, for all $j \in I_{\ell_1^*} \setminus \{k\}$, we must have $\LCP(V_k, V_j) \le \ell_2^* - 1$. This implies that the maximum LCP for this specific $k$ is bounded:
    \[ \max_{j \in I_{\ell_1^*} \setminus \{k\}} \LCP(V_k, V_j) \le \ell_2^* - 1. \]
    Since $f(\ell_1^*)$ is defined as the \emph{minimum} of such maxima over all $i \in I_{\ell_1^*}$, it follows that $f(\ell_1^*) \le \ell_2^* - 1\implies \ell_2^* \ge f(\ell_1^*) + 1$. Thus, the cost $\ell_1^* + \ell_2^*$ is at least $\ell_1^* + f(\ell_1^*) + 1$. Minimizing this over all possible $\ell_1$ yields the global optimum.
\end{proof}

To efficiently implement \Cref{lem:med:prefix-formula}, we use the following data structure. For every $(U_i, V_i) \in \mathcal{S}$, we define $\maxLCP[i] := \max_{j \ne i,\, |U_j| \ge |U_i|} \LCP(V_i, V_j)$, with $\maxLCP[i] = 0$ if no such $j$ exists. We first show how to compute this array, and then how it is used to solve \PairSUS.

\begin{fact}[Folklore] \label{lem:min-lcp}
    For any three strings $S_1$, $S_2$ and $S_3$ with $S_1 \le S_2 \le S_3$ in the lexicographical order, we have $\LCP(S_1,S_3) = \min \{\LCP(S_1,S_2), \LCP(S_2,S_3)\}$.    
\end{fact}

\begin{lemma}
    \label{lem:med:prefix-max-lcp}
    The array $\maxLCP$ can be computed in $\cO(N)$ time.
\end{lemma}
\begin{proof}
    By traversing $\mathcal{T}(\mathcal{F}_2)$, we 
    obtain the list of $V$ strings (the second components of $\mathcal S$) in lexicographical order in $\cO(N)$ total time. For convenience, 
    let us denote the \emph{sorted} list by
    $V_1,\ldots,V_N$. In addition to the list, we can also output the LCP values $\LCP(V_1,V_2),\ldots,\LCP(V_{N-1},V_N)$ in the same traversal. 
    For finding arbitrary LCP values, we use a data structure for answering range minimum queries over the latter list of $N$ LCP values. The data structure is constructed in $\cO(N)$
    time and it can answer queries in $\cO(1)$ time~\cite{DBLP:conf/latin/BenderF00}.

    From \cref{lem:min-lcp}, for any triplet $V_i, V_j, V_k$ with $|U_j|,|U_k| \ge |U_i|$ and either $V_i \le V_j \le V_k$ or $V_i \ge V_j \ge V_k$, we have $\LCP(V_i, V_k) \le \LCP(V_i, V_j)$. Repeatedly applying this argument along the lexicographical order shows that, among all strings $V_j$ with $|U_j| \ge |U_i|$, the maximum value of $\LCP(V_i,V_j)$ is attained by one of the closest such strings to $V_i$ in lexicographical order.

    We introduce two auxiliary arrays. For each $i\in[1,N]$, we define $\prev[i] := \max \{ j < i \mid |U_j| \ge |U_i| \}$ and $\nextt[i] := \min \{ j > i \mid |U_j| \ge |U_i|\}$; note that these values may be undefined for some entries. By the argument from the previous paragraph, we  have that $\maxLCP[i] = \max\{\LCP(V_i,V_{\prev[i]}), \LCP(V_i, V_{\nextt[i]})\}$ (where we treat the LCP value with some undefined entry as $0$). 
    
    The arrays $\prev$ and $\nextt$ can be computed in $\cO(N)$ time using a well-known algorithm (cf.~\cite{DBLP:journals/tcs/BarbayFN12}). From thereon, we compute the values of $\maxLCP$ by taking the maximum LCP of each entry with its corresponding $\prev$ and $\nextt$ entries using two $\cO(1)$-time range minimum queries. The total time is thus $\cO(N)$.
\end{proof}

\begin{lemma}
\label{lem:med:pairsus-prefix}
An instance of \PairSUS where $\mathcal{S}$ is a prefix family of size $N$ can be solved in $\cO(N)$ time.
\end{lemma}
\begin{proof}
Let $\mathcal{S}= \{(U_1, V_1), \dots, (U_N, V_N)\}$ be a prefix family. By \cref{lem:med:prefix-formula}, it suffices to minimize $\ell_1 + f(\ell_1) + 1$. 

The function $f(\ell_1)$ is non-increasing and changes only when $\ell_1$ passes some value $|U_i|$. Hence, it suffices to consider thresholds $\ell_1 = |U_i|$. We have
\[
f(|U_i|) = \min_{k:\ |U_k|\ge |U_i|} \left( \max_{j:\ |U_j|\ge |U_i|,\ j\ne k} \LCP(V_k,V_j) \right).
\]
By definition, $\maxLCP[k] = \max_{j:\ |U_j|\ge |U_k|,\ j\ne k} \LCP(V_k,V_j)$ upper-bounds the above expression for all $k$, and is tight when $|U_k| = |U_i|$. Thus,
\[
f(|U_i|) = \min_{k:\ |U_k|\ge |U_i|} \maxLCP[k].
\]

The algorithm proceeds as follows:
\begin{enumerate}
    \item Compute the array $\maxLCP$ in $\cO(N)$ time via \cref{lem:med:prefix-max-lcp}.
    \item Sort indices $i$ with respect to $|U_i|$ in the non-increasing order, compute the minima
    \[
        h_i = \min_{k:\ |U_k|\ge |U_i|} \maxLCP[k],
    \]
    and return $\min_i (|U_i| + 1 + h_i)$.
\end{enumerate}
Both steps take $\cO(N)$ time, and the statement follows.
\end{proof}

\subsubsection{Solution for \texorpdfstring{$(\alpha,\beta)$}{(alpha,beta)}-families}
Consider an instance of the \TwoFamilies problem in which $\P$ and $\Q$ are $(\alpha,\beta)$-families with $\log \beta = o(\log N)$.
Further, let $\R$ be the list obtained by sorting the pairs $\P \cup \Q$ of strings according to the lexicographical order of the second components. We aim to construct a wavelet tree on the first components of $\R$.
For each node $v$ of the wavelet tree, we denote by $\R_v$ the sublist of $\R$ whose elements have their first component in the leaf list $\Sigma_v$ of $v$.
For any sublist $\mathcal X = (U_1,V_1),\dots,(U_m,V_m)$ of~$\R$, we denote by $\LCPs(\mathcal X)$ the list $0,\LCP(V_1,V_2),\dots,\LCP(V_{m-1},V_{m})$, represented as a packed string over the alphabet $[0, \beta]$ in space $\cO(N/\log_\beta N)$.

\begin{lemma}[{\cite[Claim 4.3]{TALG-Charalampopoulos25}}]\label{lem:med:lcs}
Consider an instance of the \TwoFamilies problem in which $\P$ and $\Q$ are $(\alpha,\beta)$-families with $\log \beta = o(\log N)$.
We can construct, in $\cO(N(\alpha+\log N)/ \sqrt{\log N})$ time and $\cO(N+N\alpha/\log N)$ space, a wavelet tree of height $\cO(\alpha + \log N)$ for the first components of~$\R$ (some possibly padded with a $\$\not\in\Sigma$).

Moreover, in $\cO(N(\alpha+\log N) \log \beta / \log N)$ time and $\cO(N)$ space, we can compute a bit-vector~$G_v$ specifying the origin ($\P$ or $\Q$) of each element of $\R_v$ and the list $\L_v = \LCPs(\R_v)$, for each node~$v$ of the wavelet tree in the
BFS order, such that after computing $G_u$ and $\L_u$, for each child $u$ of a node $v$, 
$G_v$ and $\L_v$ are deleted.
\end{lemma}

The following lemma is a direct adaptation of the previous lemma, solving a \PairSUS instance in the same setting.

\begin{lemma}
\label{lem:med:wavelet}
Consider an instance of the \textsc{Shortest Unique String Pair} problem in which~$\mathcal{S}$ is an $(\alpha,\beta)$-family with $\log \beta = o(\log N)$.
We can construct, in $\cO(N(\alpha+\log N)/ \sqrt{\log N})$ time and $\cO(N+N\alpha/\log N)$ space, a wavelet tree of height $\cO(\alpha + \log N)$ for the first components of $\R$ (some possibly padded with a $\$\not\in\Sigma$). 

Moreover, in $\cO(N(\alpha+\log N) \log \beta / \log N)$ time and $\cO(N)$ space, we can compute $\L_v = \LCPs(\R_v)$, for each node $v$ of the wavelet tree in the BFS order, such that after computing $\L_u$ for each child $u$ of a node $v$, $\L_v$ is deleted. 
\end{lemma}

\begin{proof}
    The \TwoFamilies problem takes as input two sets of string pairs, $\mathcal{P}$ and $\mathcal{Q}$, contrary to the one set $\mathcal{S}$ in \PairSUS. \Cref{lem:med:lcs} constructs a wavelet tree precisely as described by this lemma (for the union $\mathcal{P} \cup \mathcal{Q}$), and additionally maintains, for every node, a bit vector which can be ignored here. Thus, by invoking the cited lemma with $\mathcal{P} := \mathcal{S}$ and $\mathcal{Q} := \emptyset$, we obtain exactly the required data structure within the stated  time and space bounds.
\end{proof}

The wavelet tree constructed in \Cref{lem:med:wavelet} uses the trie $\mathcal{T}(\mathcal{F}_1)$ as its skeleton tree. Although the wavelet tree is binary, its topology mirrors that of $\mathcal{T}(\mathcal{F}_1)$ so there is a many-to-one correspondence between its nodes and those of $\mathcal{T}(\mathcal{F}_1)$. Therefore, the attached LCP lists of suffixes also applies to the nodes of $\mathcal{T}(\mathcal{F}_1)$: for any node $v \in \mathcal{T}(\mathcal{F}_1)$, \cref{lem:med:wavelet} provides an LCP list for the suffixes corresponding to the leaf descendants of $v$. Inspect \cref{fig:med}, where a path in $\mathcal{T}(\mathcal{F}_1)$ spells a string preceding some anchor positions; the LCP list is used to determine the shortest prefix among the suffixes following that string. The two components are then combined to obtain a candidate SUS.

\begin{figure}
    \centering
    \begin{subfigure}[c]{0.35\textwidth}
        \centering
        \begin{tikzpicture}
            \draw [color=red] (0,0) -- (-0.15, -0.625) node [sloped, above, near end] {a}
                                    -- (-0.3, -1.25) node [sloped, above, near end] {r}
                                    -- (-0.5, -2) node [sloped, above, midway] {ab};
            \draw [dashed] (-0.5, -2) -- +(-0.5, -1) -- +(0.5, -1) -- +(0, 0);

            \foreach \xf/\yf/\xt/\yt in {0/0/2/-0.8, -0.15/-0.625/-1.6/-1.2, -0.3/-1.25/0.9/-1.5} {
                \draw [dotted] (\xf,\yf) -- (\xt,\yt);
                \fill (\xt,\yt) circle(1.5pt);
                \draw [dotted] (\xt, \yt) -- +(-0.5, -1) -- +(0.5, -1) -- +(0, 0);
            }

            \fill (0, 0) circle(1.5pt);
            \fill (-0.15, -0.625) circle(1.5pt);
            \fill (-0.3, -1.25) circle(1.5pt);
            \fill (-0.5, -2) circle(1.5pt);
            \node [anchor=west] at (-0.5, -2) {$v$};

            \draw[->] (-0.25, -2.25) to[out=-45,in=200] (3.5, -1.75);
        \end{tikzpicture}
    \end{subfigure}
    \begin{subfigure}[c]{0.35\textwidth}
        \centering
        \begin{tabular}{c|l|c}
            $i$ & suffix & $\mathcal{L}_v[i]$ \\
            \hline
            $0$ & \texttt{\textcolor{red}{cadabra}gd\ldots} & $(0)$ \\
            $1$ & \texttt{cadabrreb\ldots} & $6$ \\
            $2$ & \texttt{cadabrred\ldots} & $8$ \\
            $3$ & \texttt{carrotcak\ldots} & $2$ \\
            $4$ & \texttt{carrotcak\ldots} & $9$
        \end{tabular}
    \end{subfigure}

    \caption{Node $v$ of $\mathcal{T}(\mathcal{F}_1)$ corresponds to prefix $\texttt{arba}$ (left) with LCP list $\mathcal{L}_v$ (right). \cref{lem:med:lcps} gives~7 (shortest unique prefix \texttt{cadabra}, in red), yielding a \PairSUS  candidate of length $4 + 7 = 11$.}\label{fig:med}
\end{figure}
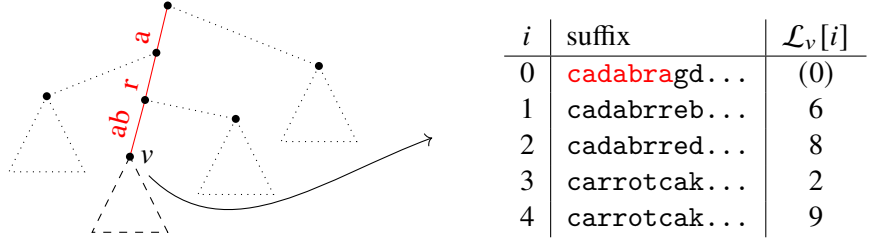

For a multiset $\mathcal{S}$ of strings, we call a lengthwise minimal string that is a prefix of exactly one element of $\mathcal{S}$ a \emph{shortest unique prefix} of $\mathcal{S}$.

\begin{lemma}
\label{lem:med:lcps}
    Let $\mathcal{S}$ be a lexicographically sorted list of $N$ strings and let $L$ be a list of $N+1$ integers such that $L[0]=L[N]=0$ and $L[i]=\LCP(\mathcal{S}[i-1],\mathcal{S}[i])$, for all $i \in [1,N)$. Then a \emph{shortest unique prefix} of $\mathcal{S}$, if one exists, can be computed in $\cO(N)$ time.
\end{lemma}
\begin{proof}
    Let $i$ be an arbitrary index in $[0,N)$. Because  $\mathcal{S}$ is sorted and by \cref{lem:min-lcp}, it follows that for any $j < i$, $\LCP(\mathcal{S}[j],\mathcal{S}[i]) \le \LCP(\mathcal{S}[i-1],\mathcal{S}[i])=L[i]$. Similarly, for any $j > i$, $\LCP(\mathcal{S}[j],\mathcal{S}[i]) \le \LCP(\mathcal{S}[i+1],\mathcal{S}[i]) = L[i+1]$. Therefore, the greatest LCP value between $\mathcal{S}[i]$ and any other string in $\mathcal{S}$ is precisely $M[i]= \max \{L[i],L[i+1]\}$. Any prefix of $\mathcal{S}[i]$ that is longer than this value (if one exists) is not a prefix of any string in $\mathcal{S} \setminus \{\mathcal S[i]\}$. Conversely, any prefix of $\mathcal{S}[i]$ of length $M[i]$ is guaranteed to be a prefix of at least one neighboring string (the one determining the maximum LCP value). Thus, the length of the shortest unique prefix for string $\mathcal{S}[i]$ is given by: \[\textsf{SUP}_i=\begin{cases}M[i]+1, & \text{if $M[i]< |\mathcal{S}[i]|$}\\
            \infty, & \text{otherwise.}
		 \end{cases}
    \]
    The length of the shortest unique prefix for the entire list $\mathcal{S}$ is the minimum of these lengths over all indices $i$, that is, $\min_{i \in [0,N)} \left\{ \textsf{SUP}_i \right\}$. This minimum can be found in $\cO(N)$ time by first computing all $M[i]$ values in $\cO(N)$ time and then finding the global minimum.
\end{proof}

\begin{example}
    Consider a lexicographically sorted list $\mathcal S = [\texttt{a}, \texttt{abacus}, \texttt{abasia}, \texttt{abate}, \texttt{abstract}]$ consisting of $N=5$ strings.
    The array $L$ for~$\mathcal S$ as defined in \cref{lem:med:lcps} is $[0,1,3,3,2,0]$.
    We compute $M[i] = \max \{L[i], L[i+1]\}$ for string $\mathcal S[i]$, for all $i\in[0,5)$, obtaining $M= [1,3, 3, 3, 2]$.  
    Then, note that, $\textsf{SUP}_0=\infty$ because $M[0]=1=|\mathcal{S}[0]|$.
    The length of the shortest unique prefix is given by $\min_{i \in [0,4)} \left\{ \textsf{SUP}_i \right\} = \min\{\infty, 4, 4, 4, 3\} = 3$, corresponding to the prefix $\texttt{abs}$ of $\texttt{abstract}$.
\end{example}

\subsubsection{Wrapping Up}
We consider separately the cases in which the candidate SUS starts within a $\tau$-run and the case when it does not. In the former case, we apply the algorithm for a prefix family, with entries based on all $\tau$-runs sharing the same structure (see \Cref{lem:med:prefix}). In the latter case, the wavelet tree algorithm for an $(\alpha, \beta)$-family suffices with entries based on positions in a string synchronizing set (see \Cref{lem:med:ab}).

\begin{lemma}
\label{lem:med:prefix}
Given an instance of \SUS, 
we can compute a SUS $P$ of $S$ in $\cO(n/\log_\sigma n)$ time 
if $|P| \in [\frac15 \log_\sigma n, 2^{\sqrt{\log n}}]$,
$\per(P) > \frac{1}{45}\log_\sigma n$,
and $\per(P[0\dd 3\lfloor\frac{1}{15}\log_\sigma n\rfloor - 1]) \le \frac{1}{45}\log_\sigma n$.
\end{lemma}
\begin{proof}
    Let $\tau := \lfloor\frac{1}{15}\log_\sigma n\rfloor$. Because $\per(P[0\dd 3\tau - 1]) \le \frac13 \tau$, we know that the SUS starts inside some $\tau$-run, and since the period of the SUS $P$ itself is by hypothesis greater, we also know that it extends beyond said $\tau$-run. We thus only have to search for a SUS around the ends of $\tau$-runs.

\proofsubparagraph{Step 1: Compute $\tau$-runs.}
    Using Lemma 2.8 from \cite{TALG-Charalampopoulos25}, 
    we compute all $\tau$-runs in $S$
    in $\cO(n/\tau)$ total time.
    
\proofsubparagraph{Step 2: Group runs by periodic suffix.}  
For each run $R$ with period $p$, consider its length-$p$ suffix $L$. We group all runs with the same suffix $L$ into a family $\mathcal{S}_L$. Since there are $\cO(n/\tau)$ runs and at most $\sigma^{\tau/3} = n^{1/45}$ distinct length-$p$ suffixes, this step takes $\cO(n/\tau) = \cO(n/\log_\sigma n)$ time in total.

\proofsubparagraph{Step 3: Construct \PairSUS instances.}  
For a run $R$ starting at position $i$ and ending at position $j$, define a string pair
\[
(U, V) := \big((S[i \dd j])^R,\, S[j+1 \dd j + 2^{\sqrt{\log n}}]\big),
\]
and add it to the corresponding family $\mathcal{S}_L$ for the suffix $L$ of $R$.  
The first component $U$ is a prefix of $(L^R)^k$ for some integer $k>0$, and thus each family $\mathcal{S}_L$ is a prefix family. By construction, the SUS must start in one of these runs, so solving \PairSUS on each family captures all candidates.

\proofsubparagraph{Step 4: Solve \PairSUS for each family.}  
By Lemma~\ref{lem:med:pairsus-prefix}, \PairSUS can be solved in linear time with respect to the size of the family. Each family has one string pair per run, so all families together take $\cO(n/\tau) = \cO(n/\log_\sigma n)$ time.

\proofsubparagraph{Step 5: Recover the SUS.}  
The SUS is obtained by taking the minimal-cost solution returned by \PairSUS for any family (if the first part of the solution is shorter than $3\tau$, then it does not have a periodic prefix -- candidate for SUS in the other case is always better). 

\proofsubparagraph{Analysis.}  Correctness follows because the SUS must start within a $\tau$-run and extend beyond it, and all possible runs and their periodic suffixes are considered. All steps---computing $\tau$-runs, grouping by suffix, constructing families, and solving \PairSUS---take $\cO(n/\tau) = \cO(n/\log_\sigma n)$ time.
\end{proof}

\begin{lemma}
\label{lem:med:ab}
Given an instance of \SUS, 
we can compute a SUS~$P$ of $S$ in $\cO(n\log\sigma/\sqrt{\log n})$ time 
if $|P| \in [\frac15 \log_\sigma n, 2^{\sqrt{\log n}}]$ 
and $\per(P[0\dd 3\lfloor\frac{1}{15}\log_\sigma n\rfloor - 1]) > \frac{1}{45}\log_\sigma n$.
\end{lemma}
\begin{proof}
    Let $\tau := \lfloor\frac{1}{15}\log_\sigma n \rfloor$, $\beta := \lfloor 2^{\sqrt{\log n}} \rfloor$, and let $\Sync$ be a $\tau$-synchronizing set of $S$ computed using \cref{the:sss} in $\cO(n/\log_\sigma n)$ time.
    For each $i \in \Sync$, define $\text{pred}(i) = \max (\{j + 1 \mid j \in \Sync, j < i\} \cup \{i - \tau + 1, 0\})$.
    We construct a multiset~$\mathcal{S}$ of string pairs containing, for each $i \in \Sync$, the pair $(S[\text{pred}(i)\dd i-1]^R,S[i\dd \min\{i+\beta,n\}))$.
    Note that $\mathcal{S}$ is a $(\tau,\beta)$-family: the first components have length at most $\tau$ and the second components have length at most $\beta$.

    Let $N := |\Sync| = \cO(n/\log_\sigma n)$.
    We construct the compacted trie $\mathcal{T}_1$ of the first components of $\mathcal{S}$ in the following way.
    Since all first components have length at most $\tau$, we can sort them in $\cO(N + \sigma^\tau) = \cO(n/\log_\sigma n + n^\frac{1}{15}) =\cO(n/\log_\sigma n)$ time using bucket sort.
    Then, using an LCE data structure on~$S$, which can be constructed in $\cO(n/\log_\sigma n)$ time~\cite{DBLP:conf/stoc/KempaK19}, we can compute the LCP values between consecutive components; with this information at hand, we can construct $\mathcal{T}_1$ in $\cO(N)$ time~\cite{DBLP:conf/cpm/KasaiLAAP01}.
    The compacted trie $\mathcal{T}_2$ of the second components
    can be constructed in the same way after sorting
    the suffixes starting at the positions in $\Sync$ using~\cite[Theorem 4.3]{DBLP:conf/stoc/KempaK19} in $\cO(|\Sync|) = \cO(n/\log_\sigma n)$ time.
    The set $\mathcal{S}$ and the two tries $\mathcal{T}_1$ and $\mathcal{T}_2$ form an instance of the \PairSUS problem.

    \proofsubparagraph{Correctness.}  We first show that solving this \PairSUS instance yields the desired SUS.
    Suppose that there exists a substring $R$ of $S$ of length $\ell \in [3\tau,\beta]$ and period $\per(R[0\dd 3\tau - 1]) > \frac13 \tau$. Then, by the synchronizing property, there exists some $j \in [0,\tau)$ such that every occurrence of $R$ has an anchor from $\Sync$ at its $j$-th position.
    We can decompose $R$ as $R= R_1\cdot R_2$, where $|R_1|=j$. There is a bijection between occurrences of $R$ in $S$ and the pairs $(P_1,P_2) \in \mathcal{S}$ such that $R_1^R$ is a prefix of $P_1$ and $R_2$ is a prefix of $P_2$. If $R$ is unique in $S$, there is exactly one such pair $(P_1,P_2)$. For any other pair $(Q_1,Q_2) \in \mathcal{S}$, it must hold that $\LCP(P_1,Q_1) < |R_1|$ or $\LCP(P_2, Q_2) < |R_2|$. Thus, the minimum $\ell_1+\ell_2 = |R_1|+|R_2|$ defines a feasible solution to \PairSUS. Conversely, if the optimal  \PairSUS solution is the pair $(P_1,P_2)$ with lengths $\ell_1$, $\ell_2$, we locate the anchor corresponding to $(P_1,P_2)$ in $S$, and output the substring of length $(\ell_1+\ell_2)$ surrounding it; this substring must be unique by the same reasoning.

   \proofsubparagraph{Construction and analysis.} We apply \cref{lem:med:wavelet} to construct a wavelet tree for the first components of $\mathcal R$ (the list obtained by sorting $\mathcal S$ according to the lexicographical order of the second components) and to compute the corresponding LCP lists for all nodes. Each node corresponds to a common prefix of length $\ell_1$ for a subset of the  first components. Within each node, we invoke \cref{lem:med:lcps} to find the shortest unique suffix extension of length $\ell_2$; concatenating  the two parts yields a candidate SUS.

We now analyze the running time of the described algorithm.
Constructing the wavelet tree takes $\cO(N(\tau+\log N)/\sqrt{\log N}) = \cO(n\log \sigma/\sqrt{\log n})$ time, and computing the LCP lists takes $\cO(N+N\tau \log \beta / \log N) = \cO(n/\sqrt{\log n})$ time, by \cref{lem:med:wavelet}. Processing  each list to extract SUS candidate takes linear time in its  length, so the overall running time of the algorithm is $\cO(n\log \sigma/\sqrt{\log n})$.
\end{proof}

\begin{corollary}
    \label{lem:med}
    Given an instance of \SUS, 
    we can compute a SUS $P$ of $S$ in $\cO(n\log\sigma/\sqrt{\log n})$ time 
    if it has length $|P| \in [\frac15 \log_\sigma n, 2^{\sqrt{\log n}}]$ 
    and period $p > \frac{1}{45}\log_\sigma n$.
\end{corollary}
\begin{proof}
Let $\tau=\lfloor\frac1{15}\log_\sigma n\rfloor$. Either $\per(P[0\dd 3\tau-1])\le \frac13\tau$
and \cref{lem:med:prefix} applies, or
$\per(P[0\dd 3\tau-1])> \frac13\tau$
and \cref{lem:med:ab} applies.
\end{proof}

\subsection{SUS as a Skyline Problem}\label{sec:skyline}

In this section, we introduce and solve a simple geometric problem that serves as a subroutine for computing a SUS in both the long aperiodic and the periodic cases.

\begin{definition}[Domination in $\mathbb{Z}_{\ge 0}^2$]
    Let $p_1=(x_1,y_1)$ and $p_2=(x_2,y_2)$ be two points in~$\mathbb{Z}_{\ge 0}^2$. We say that $p_1$ is \emph{dominated by} $p_2$ if $x_1 \le x_2$ and $y_1 \le y_2$.
\end{definition}

\begin{definition}[Shadow]
    Let $P$ be a multiset of points in $\mathbb{Z}_{\ge0}^2$. The \emph{shadow} of a point $p \in P$ is the set of points in $\mathbb{Z}_{\ge0}^2$ that are dominated by $p$ but not dominated by any $p' \in P \setminus \{p\}$.
\end{definition}

\begin{definition}[Skyline]
    Let $P$ be a multiset of points in $\mathbb{Z}_{\ge0}^2$. The \emph{primary skyline} of $P$ is the union of the shadows of all points in $P$.
\end{definition}

\begin{example}
    The primary skyline of $\{(7,3), (7,3)\}$ is the empty set.
\end{example}

In the following, we formalize the \SkylineProblem problem, explain its relevance to finding a SUS of $S$, and present a linear-time solution; see \Cref{fig:MSP} for an example.

\defproblem{\SkylineProblem}{A multiset $P$ of points in $\mathbb{Z}_{\ge0}^2$.}{A point $(x,y)\in\mathbb{Z}_{\ge0}^2$ that lies in the primary skyline of $P$ and minimizes $x + y$, if one exists.}

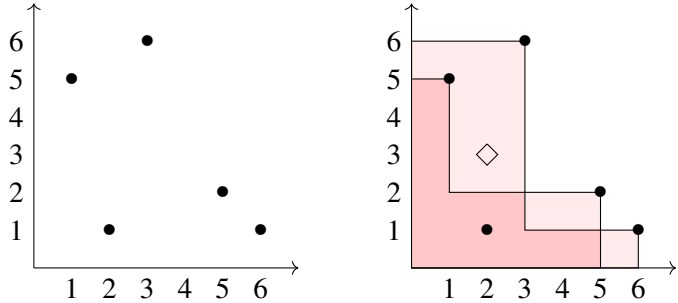
\begin{figure}[htpb!]
    \centering
    \begin{subfigure}{0.22\textwidth}
        \centering
        \begin{tikzpicture}[x=5mm, y=5mm]
            \def\PicSize{6};
            \def\Points{6/1, 5/2, 2/1, 3/6, 1/5};
            
            \foreach \x/\y in \Points {
                \node at (\x, \y) {$\bullet$};
            }
            \draw[->] (0, 0) -> (\PicSize + 1, 0);
            \draw[->] (0, 0) -> (0, \PicSize + 1);
            \foreach\i in {1,...,\PicSize} {
                \node[anchor=east] at (0, \i) {$\i$};
                \node[anchor=north] at (\i, 0) {$\i$};
            }
        \end{tikzpicture}
    \end{subfigure}
    \hspace{1cm}
    \begin{subfigure}{0.22\textwidth}
        \centering
        \begin{tikzpicture}[x=5mm, y=5mm]
            \def\PicSize{6};
            \def\Points{6/1, 5/2, 2/1, 3/6, 1/5};

            \filldraw[fill=red!8]
                (6, 0) -- (6, 1) -- (5, 1) -- (5, 2) -- (3, 2) -- (3, 6) -- (0, 6) -- (0,0) -- cycle;

            \filldraw[fill=red!22]
                (0, 0) -- (5, 0) -- (5, 1) -- (3, 1) -- (3, 2) -- (1, 2) -- (1, 5) -- (0, 5);
            
            \foreach \x/\y in \Points {
                \node at (\x, \y) {$\bullet$};
            }

            \draw[->] (0, 0) -> (\PicSize + 1, 0);
            \draw[->] (0, 0) -> (0, \PicSize + 1);
            \foreach\i in {1,...,\PicSize} {
                \node[anchor=east] at (0, \i) {$\i$};
                \node[anchor=north] at (\i, 0) {$\i$};
            }

			\draw (2,3) +(0, 4pt) -- +(-4pt, 0) -- +(0, -4pt) -- +(4pt, 0) -- +(0, 4pt);
        \end{tikzpicture}
    \end{subfigure}
    \caption{A set $P$ (left) and its primary skyline (right).
    The union of the shadows of the points in $P$ consists of all points within the shaded polygon but outside its heavily shaded part---the points in the heavily shaded part are dominated by multiple points from $P$.
    For this instance of \SkylineProblem, the output would be point $(2,3)$, shown with a rhombus.}\label{fig:MSP}
\end{figure}

\begin{lemma}
    \label{lem:skyline-alg}
    Any instance of \SkylineProblem can be solved in $\cO(|P|)$ time if the points in $P$ are given as a list sorted with respect to one of the two coordinates.
\end{lemma}
\begin{proof}
    We first establish two claims.

    \begin{claim}\label{lem:skyline-alg:hor}
        Suppose that the primary skyline is nonempty and let $(x^\star,y^\star)$ be a point in the primary skyline with minimal sum of coordinates. Then, either $x^\star=0$ or there exists a point $(x,y) \in P$ with $x = x^\star - 1$.
    \end{claim}
    \begin{proof}
       Assume, toward a contradiction, that the claim is false; that is, that we have $x^\star > 0$ and there is no point $(x,y) \in P$ with $x = x^\star - 1$. 
       Let $p^\star = (x^\star, y^\star)$.
       By the definition of the primary skyline, $p^\star$ must be dominated by exactly one point in $P$.
       Now consider the point $p_{\text{left}} = (x^\star - 1, y^\star)$, which, due to the minimality of the coordinate-wise sum of $p^\star$, must be dominated by at least two points in $P$. Therefore, there exists a point $p' \in P$ that dominates $p_{\text{left}}$ and does not dominate $p^\star$.
       Let $p' =(x',y')$. Since $p'$ does not dominate $p^\star$, we must have either $x' < x^\star$ or $y' < y^\star$. We obtain a contradiction in each case:
       \begin{itemize}
        \item If $x' < x^\star$, then by our assumption that no point $(x,y) \in P$ has $x=x^\star - 1$, we must have $x' \le x^\star - 2$. Hence, $x' < x^\star - 1$, and thus $p'$ does not dominate $p_{\text{left}}$, a contradiction.
        \item If $y' < y^\star$, then $p'$ clearly does not dominate $p_{\text{left}}$, a contradiction.\qedhere
        \end{itemize}
    \end{proof}
    
    \begin{claim}
        \label{lem:skyline-alg:ver}
        Suppose that the primary skyline is nonempty and let $(x^\star,y^\star)$ be a point in the primary skyline with minimal sum of coordinates.
        Let $P' = \{(x,y) \in P \mid x \ge x^\star \}$ and let $p_{\text{unique}}$ be the unique point of $P'$ that dominates $(x^\star,y^\star)$.
        If $|P'| = 1$, we have $y^\star = 0$.
        Else, we have $y^\star = y' + 1$, where $y'$ is the maximum $y$-coordinate of a point in $P' \setminus \{p_{\text{unique}}\}$. 
    \end{claim}
    \begin{proof}

         By definition, $p^\star:=(x^\star,y^\star)$ is dominated by exactly one point, namely $p_{\text{unique}}$.
         Since points are only dominated by points weakly to their right, we have that $p_{\text{unique}}$ must be in~$P'$.
         We distinguish between two cases:
         
         \proofsubparagraph{Case 1: $|P'| = 1$.}
         We have $P' = \{p_{\text{unique}}\}$.
         If $y^\star > 0$, then the point $p_{\text{down}} = (x^\star, y^\star - 1)$ is in $\mathbb{Z}_{\ge0}^2$ and the only point of $P$ that dominates it is $p_{\text{unique}}$.
         Thus, $p_{\text{down}}$ lies in the primary skyline, but its coordinate sum $x^\star + y^\star - 1$ is smaller than that of $p^\star$, contradicting the minimality of $p^\star$.
         Hence, $y^\star = 0$.
         
        \proofsubparagraph{Case 2: $|P'| \ge 2$.}
        Let $P_{\text{other}} = P' \setminus \{p_{\text{unique}}\}$, and let $y'$ be the maximum $y$-coordinate of a point in $P_{\text{other}}$. Since~$p^\star$ is dominated uniquely by $p_{\text{unique}}$, every point $(x'',y'') \in P_{\text{other}}$ satisfies $y'' < y^\star$. Hence, $y' < y^\star$. Conversely, since $p_{\text{unique}}$ dominates $p^\star$, its $y$-coordinate is at least $y^\star$. Therefore, the point $(x^\star,y^\star)$ is dominated by exactly one point in $P'$ if and only if $y' < y^\star$. By the minimality of $p^\star$, we have $y^\star = y' + 1$.
    \end{proof}

    \proofsubparagraph{Algorithm.}
    We assume, without loss of generality, that the points in $P$ are sorted by $x$-coordinate in non-decreasing order. We scan $P$ from right to left, while maintaining $y_{\max}$ and $y'$, the two largest $y$-coordinates encountered among processed points (where $y_{\max}\geq y'$). We initialize $y_{\max}$ and $y'$ to $-1$, indicating that no processed point has contributed a valid $y$-coordinate yet.
    
    We iterate through $P$ by grouping points with the same $x$-coordinate. Let $x_i$ denote the current $x$-coordinate of the group currently being processed, and let $P_{\text{curr}}=\{(x,y)\in P \mid  x=x_i\}$ be the set of points at this coordinate.
    Before processing the points in $P_{\text{curr}}$, we evaluate the candidate $x$-coordinate  $x^\star= x_i +1$. Because we have scanned from right to left, all points $(x,y)\in P$ such that $x\geq x^\star$ have already been processed. Let $P'$ denote this set of previously processed points. We determine the corresponding $y^\star$ following \cref{lem:skyline-alg:ver}:
    \begin{itemize}
        \item If $|P'| = 1$, we set $y^\star = 0$.
        \item If $|P'| \ge 2$ and $y_{\max} > y'$, we set $y^\star = y' + 1$. 
        \item Otherwise (if $|P'|=0$ or $y_{\max} = y'$), the primary skyline does not intersect the set $\{(x^\star,y) \mid y \in \mathbb{Z}_{\ge0}\}$, so  we skip this candidate.
    \end{itemize}
    After evaluating $y^\star$, we update $y_{\max}$ and $y'$ using the $y$-values in $P_{\text{curr}}$ so that they remain the two largest $y$-coordinates among all processed points.
    
    Finally, after all points in $P$ have been processed, we perform a last check for the candidate $x^\star=0$ using the final values of $y_{\max}$ and $y'$. The algorithm maintains the candidate $(x^\star, y^\star)$ that minimizes $x^\star + y^\star$ and returns it as a witness. Each point of $P$ is processed exactly once using $\cO(1)$ simple operations, thus the algorithm runs in $\cO(|P|)$ time.
\end{proof}

\subparagraph{Intuition for application to SUS.}
In the long aperiodic case (\Cref{sec:long}) and in the periodic case (\Cref{sec:periodic}), we show that a SUS of $S$ can be found by solving several instances of the \SkylineProblem problem. The core idea is to transform the SUS problem into a geometric one. Let us consider one such instance. We first identify a set $\mathcal{S}$ of carefully-selected substrings of $S$, each anchored
around a \emph{common fragment} $F=S[i \dd j]$. Each substring $U \in \mathcal{S}$ is then represented by an integer point $(x,y)$, where $x$ and $y$ are the lengths of the extensions of $U$ to the left and the right, respectively, relative to the common fragment $F$.

\begin{example}
    Let $S=\texttt{abracadabra}$, $F=S[4\dd 6]=\texttt{cad}$, and $\mathcal{S}=\{\texttt{bracada},\texttt{cadabra}\}$. The substrings are represented as points by the lengths of their extensions relative to $F$:\\  $\texttt{bracada}$ $\mapsto (3,1)$ (since \texttt{bra} is of length 3 and \texttt{a} is of length 1),  and  $\texttt{cadabra}$ $\mapsto (0,4)$ (since the left extension is empty and \texttt{abra} is of length 4).
\end{example}

Since strings in $\mathcal{S}$ are substrings of $S$ anchored around a common fragment, each point maps to one substring of $S$.
Then, a string $U \in \mathcal{S}$ is a substring of string $V \in \mathcal{S}$ if and only if the point representing~$U$ is dominated by the point representing $V$. To find a candidate SUS in $\mathcal{S}$, we look for a minimal extension anchored around~$F$ that is contained within exactly one substring from $\mathcal{S}$. This is equivalent to solving a \SkylineProblem instance: finding an integer point $(x^\star,y^\star)$ that lies in the primary skyline of the points representing $\mathcal{S}$ that minimizes
$x^\star+y^\star$. The total length of the resulting SUS is then $|F| + x^\star + y^\star$. 

\subsection{Long Aperiodic Case}\label{sec:long}

This section addresses the computation of a SUS in the case when
$\ell \ge \log^4 n$ and $p > \frac{1}{9}\log^4 n$.
Let us start with a high-level overview of our solution for this case. We construct a string synchronizing set consisting of \emph{anchor} positions in $S$ that enable the identification of identical sufficiently long aperiodic patterns. These anchors are then used to construct two tries whose root-to-leaf paths represent occurrences of these patterns. SUSs are found by locating the shortest paths that occur only once. We achieve this by formulating several instances of the \SkylineProblem problem based on these tries.

We apply \cref{the:sss} on $S$ with $\tau := \lfloor\frac13 \log^4 n \rfloor$ and denote the resulting $\tau$-synchronizing set by $\Sync$. The following proposition 
is crucial: if a sufficiently long, aperiodic pattern occurs twice in the string, then every anchor within one occurrence must have a corresponding anchor at the same relative position in the other occurrence. Therefore, to find a unique pattern, we must locate an anchor whose surrounding fragments---those immediately preceding and succeeding it---do not occur around any other anchor and have minimal total length.

\begin{prop}[follows by \Cref{def:SSS}]\label{prop:sss}
   Suppose that a substring $X$ of $S$ with length $|X|=m \ge 3\tau$ and period $\per(X) > \frac13 \tau$ occurs at distinct positions $i$ and $j$ in $S$.
   Then for all $q \in [0, m - 2\tau]$, we have $i + q \in \Sync$ if and only if $j + q \in \Sync$.
\end{prop}

   We construct two tries $T$ and $T^R$. For every anchor position $i \in \Sync$, we insert $S[i\dd n)$ into $T$ and $S[0\dd i)^R$ into $T^R$. In both tries, the leaves are labeled with their corresponding anchor positions. For efficiency, we implement them as \emph{compacted} tries. For any explicit node $v$, let $L(v)$ denote the set of leaf labels in the subtree rooted at $v$; recall that $\str(v)$ denotes the concatenation of the edge labels from the root to $v$ and $\sd(v)$ denotes the string depth of $v$. 
   Note that these definitions extend naturally to implicit nodes.
   If an occurrence of a pattern~$P$ contains an anchor $i \in \Sync$, then there exist nodes $u$ in $T^R$ and $v$ in $T$ such that $P = \str(u)^R\cdot \str(v)$ and $i \in L(u) \cap L(v)$. Moreover, the size $|L(u)\cap L(v)|$ equals the number of occurrences of $P$, provided that $P$ is sufficiently long and aperiodic. We formalize our task as follows; for convenience, for the remainder of this section, we assume that the parent of the root node is itself.
   
    \defproblem{\TwoTreesProblem}{Two rooted trees $T_1$ and $T_2$, each with size at most $N$ and with leaves uniquely labeled from $[0,N)$, a weight function $w$ with range $\mathbb{Z}_{\geq 0}$ where $w(u) > w(v)$ for every strict ancestor~$v$ of every node $u$, and an integer~$k$.}{A pair $(u,v)$ of nodes $u$ in $T_1$ and $v$ in $T_2$ minimizing $w(\textsf{parent}(u)) + \max(w(\textsf{parent}(v)),k-1)$, subject to $|L(u)\cap L(v)|=1$ and $w(v)\ge k$ (if they exist).}

\begin{lemma}
    \label{lem:long:two-trees-reduction}
    Consider a \SUS instance and let $\tau := \lfloor \frac13 \log^4 n \rfloor$. If $S$ has a SUS of length $\ell \ge 3\tau$ and period $p > \frac13 \tau$, we can reduce its computation in $\cO(n/\log_\sigma n)$ time to an instance of \TwoTreesProblem with $N =\cO(n/\tau)$ and $k = 2\tau$.
\end{lemma}
\begin{proof}
    We begin by constructing a $\tau$-synchronizing set $\Sync$ of $S$.
    Using this set, we construct two compacted tries $T$ and $T^R$.
    For every anchor $i \in \Sync$, we insert the suffix $S[i\dd n)$ into $T$ and the reversed prefix $S[0\dd i)^R$ into $T^R$, labeling the corresponding leaves with the anchor $i$ in both tries. 
    We make all implicit nodes in $T$ at string depth $k-1$ explicit.
    Hence, we have $N=\cO(|\Sync|)=\cO(n/\tau)=\cO(n/\log^4 n)$.

    Any fragment $P$ in $S$ that contains at least one anchor $i\in \Sync$ can be decomposed into $P = \str(u)^R \cdot \str(v)$, where $u$ is a node in $T^R$ and $v$ is a node in $T$, both of which can be implicit.
    If $P$ has length at least~$3\tau$ and period greater than $\frac13 \tau$, then by the definition of $\tau$-synchronizing sets, all occurrences of~$P$ have anchors at the same relative positions. Specifically, any occurrence of $P$ starting at position $j$ must have $j+q \in \Sync$ if and only if the original occurrence at position $i$ had an anchor at the same offset (\Cref{prop:sss}).
    Thus, if the node pair $(u,v)$ represents the fragment $P$, the number of occurrences of $P$ is exactly $|L(u) \cap L(v)|$.
    To find a SUS, we must find a pair  $(u,v)$ of (possibly implicit) nodes such that $|L(u) \cap L(v)| = 1$ (uniqueness), $\sd(v) \ge 2\tau$ (to satisfy the anchor offset property), and $\sd(u)+\sd(v)$ is minimized.

    To do so, we set the weights of nodes in the two trees as follows.
    In both tries, we set the weight of the root node to $0$ and assign to each other node $u$ a weight $w(u) := \sd(u)$. We thus create an instance of \TwoTreesProblem with  $k=2\tau$.
    If the optimal node pair $(u,v)$ satisfies $w(\textsf{parent}(u)) + w(\textsf{parent}(v)) +2 < 3\tau$, we conclude that $S$ does not have a SUS of length $\ell \ge 3\tau$ and period $p > \frac13 \tau$.
    (Note that since we have made all nodes with string depth $k-1$ explicit, for any node $v$ with $w(v) \geq k$, $v$'s parent has weight at least $k-1$, and hence $\max (w(\textsf{parent}(v)),k-1) = w(\textsf{parent}(v))$.)
    Otherwise, we return as a SUS the string  $\str(u)^R[0 \dd w(\textsf{parent}(u))] \cdot \str(v)[0 \dd w(\textsf{parent}(v))]$---the indices of the
    occurrence of this string in $S$ can be inferred from the 
    common leaf label in the subtrees of $u$ and $v$.

    The time complexity depends on the construction of the $\tau$-synchronizing set, which requires $\cO(n/\log_\sigma n)$ time using \Cref{the:sss}, and the construction of the tries.
    For the tries, we first construct an LCE data structure over $S$ in $\cO(n/\log_\sigma n)$ time~\cite{DBLP:conf/stoc/KempaK19} supporting $\cO(1)$-time LCP queries. We sort the $N$ suffixes of $S$ in $\cO(N\log N)$ time using merge sort, performing each comparison with an LCP query and a letter comparison.
    Given the sorted list of suffixes and the LCE data structure, the tries can then be constructed in $\cO(N)$ time~\cite{DBLP:conf/cpm/KasaiLAAP01}.
    The total time for constructing the tries is $\cO(n/\log_\sigma n)+\cO(N \log N) \subseteq \cO(n/\log_\sigma n)$.
\end{proof}

\subsubsection{Solving the \TwoTreesProblem Problem}

We solve \TwoTreesProblem as follows.
We first decompose the trees $T_1$ and $T_2$ into \emph{heavy paths}.
For each pair of heavy paths, one from $T_1$ and one from $T_2$, 
we then construct a \SkylineProblem instance, which we solve in linear time using \cref{lem:skyline-alg}.

\subparagraph{Heavy-light decomposition.} We first recall the widely-used heavy-light decomposition~\cite{DBLP:journals/jcss/SleatorT83}.

\begin{definition}[Heavy-Light Decomposition~\cite{DBLP:journals/jcss/SleatorT83}]
Consider a rooted tree $\mathcal T$.
We obtain a \emph{heavy-light decomposition} of $\mathcal T$ by marking each edge as either heavy or light as follows. For every internal node of $\mathcal T$, the outgoing edge leading to the child with the largest number of descendants is marked as heavy, while all other outgoing edges are marked as light; ties are resolved arbitrarily.
    A maximal path of heavy edges is a \emph{heavy path}. 
\end{definition}

A heavy-light decomposition can be constructed in linear time~\cite{DBLP:journals/jcss/SleatorT83}. The following fact holds for any heavy-light decomposition:

\begin{fact}[Lemma 1~\cite{DBLP:journals/jcss/SleatorT83}]
    \label{lem:hp-height}
    In a tree with $N$ leaves, any root-to-leaf path intersects  at most $\log N$ heavy paths in the decomposition.
\end{fact}

Furthermore, we can construct a \emph{heavy-path tree} that encodes the ancestral relations among all heavy paths and leaves. This is done by contracting every heavy edge, such that all remaining edges are light.  In this auxiliary structure, each node corresponds to a contracted heavy path or  a leaf. By \cref{lem:hp-height}, the heavy-path tree has height at most $\log N$, which allows for efficient traversal across the original tree structure.

\subparagraph{A reduction using pairs of heavy paths.}
We first note that any internal node of a tree belongs to exactly one heavy path in its heavy-light decomposition.
We compute the heavy-light decompositions for $T_1$ and $T_2$ and consider every pair of  heavy paths $h_1$ from~$T_1$ and $h_2$ from $T_2$, provided that they share at least one leaf label.

For a heavy path $h$, let $L(h)$ denote the set of all leaves descending from the root of $h$.
We define the function $d_h : L(h) \to \mathbb{Z}_{\geq 0}$, where for some leaf $\ell \in L(h)$, $d_h(\ell)$ denotes the weight of $\ell$'s lowest ancestor within~$h$.
Namely, $d_h(\ell)$ denotes the maximum weight of a node on $h$ that is an ancestor of leaf $\ell$.
With these definitions, we have the following observation:

\begin{observation}
\label{obs:heavy-path-lca}
    Given a heavy path $h$, a node $v$ on $h$, and a leaf $\ell$ descending from the root of $h$, we have that $\ell \in L(v)$ if and only if $w(v) \le d_h(\ell)$.
\end{observation}

To solve \TwoTreesProblem, we show that it suffices to solve a \SkylineProblem instance for every pair of heavy paths $(h_1,h_2)$ sharing at least one leaf label.
By \cref{lem:hp-height}, each leaf belongs to at most $\log N$ sets $L(h)$ per tree. We can thus construct each subset $L(h_1)\cap L(h_2)$ by enumerating, for each leaf label $\ell$, all $\log^2 N$ pairs of heavy paths above it.

\begin{lemma}
    \label{lem:long:two-trees-algo}
    Any instance of \TwoTreesProblem can be solved in $\cO(N \log^3 N)$ time.
\end{lemma}
\begin{proof}
    We construct the heavy-light decompositions of $T_1$ and $T_2$ in $\cO(N)$ time~\cite{DBLP:journals/jcss/SleatorT83}. This process partitions each tree into a set of disjoint heavy paths. 

    For every leaf label $\ell \in [0,N)$ present in both $T_1$ and $T_2$, we identify all pairs of heavy paths $(h_1,h_2)$ such that $h_1$ lies on the root-to-leaf path in the first tree and $h_2$ lies on the root-to-leaf path in the second tree.
    Since any root-to-leaf path intersects at most $\cO(\log N)$ heavy paths (\cref{lem:hp-height}), there are $\cO(\log^2 N)$ such pairs for each of the $N$ labels. 
    For each pair, we generate a tuple $(h_1,h_2,d_{h_1}(\ell),d_{h_2}(\ell))$, where $d_{h_1}$ and $d_{h_2}$ are retrieved from the weights of the light-edge endpoints.
    All tuples are generated in $\cO(N \log^2 N)$ time in total and are stored in one list. We sort the list, using the heavy path identifiers $(h_1,h_2)$ as primary keys and the value $d_{h_1}(\ell)$ as the secondary key. This ensures that all tuples corresponding to the same pair $(h_1,h_2)$ appear consecutively, and are then ordered by $d_{h_1}(\ell)$. Using merge sort, this step takes $\cO(N \log^2 N \cdot \log(N \log^2 N))=\cO(N\log^3 N)$ time.
    Let us denote the sublist corresponding to the pair $(h_1,h_2)$ of heavy paths by $\mathcal{H}(h_1,h_2)$.

    \begin{claim}\label{lem:long:skyline}
    Consider an instance of \TwoTreesProblem, in which the output nodes are restricted to a given pair of heavy paths $h_1$ and $h_2$ from $T_1$ and $T_2$. Given the list $\mathcal{H}(h_1,h_2)$, we can reduce this instance in $\cO(g)$ time to an instance of the \SkylineProblem problem over a multiset of points of size $\cO(g)$, where $g = |\mathcal{H}(h_1,h_2)|$.
    \end{claim}
    \begin{claimproof}
    Recall that in the definition of \TwoTreesProblem, the integer $k$ is given as a minimum weight on one of the returned nodes.
    Let $\mathcal{L} = L(h_1) \cap L(h_2)$. Assume that $\mathcal{L} \ne \emptyset$; otherwise, the instance has no solution.
    We wish to find nodes $u\in h_1$ and $v\in h_2$ minimizing $w(\textsf{parent}(u)) + \max(w(\textsf{parent}(v)),k-1)$ such that $|L(u) \cap L(v)| = 1$ and $w(v) \ge k$.
    By \cref{obs:heavy-path-lca}, this reduces to finding the minimum $(x_1,x_2) \in \mathbb{Z}_{\ge 0}^2$ such that exactly one leaf $\ell \in \mathcal{L}$ satisfies $d_{h_1}(\ell) \ge x_1$ and $d_{h_2}(\ell) \ge x_2$.
    This is equivalent to solving a \SkylineProblem on the multiset $P' = \{(d_{h_1}(\ell),d_{h_2}(\ell)) \mid \ell \in \mathcal{L}\}\sqcup\{(\infty,k-1),(\infty,k-1)\}$, thus finding a solution $(x^\star, y^\star)$ with minimal $x^\star + y^\star$---the extra point that we insert in the multiset twice ensures that if the primary skyline is not empty, then $y^\star\ge k$.
    Given $\mathcal{H}(h_1,h_2)$, we can construct this instance in $\cO(g)$ time. Finally, to obtain a solution to \TwoTreesProblem, we select the nodes $u$ in $h_1$ and $v$ in $h_2$ satisfying $w(\textsf{parent}(u)) +1 = x^\star$  and either $w(\textsf{parent}(v)) + 1 = y^\star$ or $w(\textsf{parent}(v)) + 1 \leq k = y^\star \leq w(v)$.
    \end{claimproof}

    We apply \cref{lem:long:skyline} to each pair of heavy paths that share at least one leaf label, and solve each instance of the \SkylineProblem problem using \cref{lem:skyline-alg} in time linear in the number of points.
    Since the total number of tuples across all pairs of heavy paths is $\cO(N \log^2 N)$, the total time required for applications of \cref{lem:long:skyline} and \cref{lem:skyline-alg} is $\cO(N \log^2 N)$.
    We maintain the global minimum value of $w(\textsf{parent}(u)) + \max(w(\textsf{parent}(v)),{k-1})$ found across all pairs of heavy paths and return a witness pair of nodes as the final solution.
\end{proof}

\subsubsection{Wrapping Up}
The final complexity of the long aperiodic case is determined by the combination of the reduction to the \TwoTreesProblem problem and the subsequent application of the heavy-path-based \SkylineProblem algorithm. 

\begin{lemma}
    \label{lem:long}
Given an instance of \SUS, 
we can compute a SUS of $S$ in $\cO(n/\log_\sigma n)$ time 
if it has length $\ell \ge \log^4 n$ and period $p > \frac19 \log^4 n$.
\end{lemma}
\begin{proof}
    Using \cref{lem:long:two-trees-reduction}, we reduce the long aperiodic case to an instance of the \TwoTreesProblem problem. The number of leaf labels (anchors) is $N =  \cO(|\Sync|) = \cO(n / \log^4 n)$.
    As established in the reduction, constructing the synchronizing set, the LCE data structure, and the two compacted tries $T$ and $T^R$ takes
    $\cO(n/\log_\sigma n)$ time.
    We then solve the resulting \TwoTreesProblem instance using the algorithm from \cref{lem:long:two-trees-algo}. The time required is $\cO(N\log^3 N)$ and as $N$ is in $\cO(n/\log^4 n)$, we obtain a running time of $\cO(n/\log n)$.
    Therefore, both the reduction and the solver fit within the target time bound of $\cO(n / \log_\sigma n)$. The algorithm identifies a pair $(u, v)$ minimizing $w(\textsf{parent}(u)) + \max(w(\textsf{parent}(v)),{k-1})$ subject to the uniqueness and length constraints. Since the result of the \TwoTreesProblem problem provides the minimal unique extension for any pattern overlapping an anchor, the resulting substring is a valid SUS for the long aperiodic case.
\end{proof}

\subsection{Periodic Case}\label{sec:periodic} 

In this section, we address the SUS computation in highly periodic substrings of $S$.  Specifically, we consider substrings of $S$ of medium length with period at most $\frac{1}{45}\log_\sigma n$, as well as long substrings of $S$ with period at most $\frac{1}{9}\log^4 n$. These two cases are handled using $\tau$-runs with $\tau=\lfloor \frac{1}{15} \log_{\sigma}n \rfloor$ and $\tau = \lfloor \frac{1}{3}\log^4 n \rfloor$, respectively. 

For the first group of $\tau$-runs (where $\tau=\lfloor \frac{1}{15} \log_{\sigma}n \rfloor$), we employ their standard Lyndon representation (\Cref{def:canonical}).
These runs can be efficiently computed and grouped by means of \cref{lem:per:find-roots}. 
Unfortunately, \cref{lem:per:find-roots} is not applicable when $\tau = \lfloor \frac{1}{3}\log^4 n \rfloor$.
To deal with this, we utilize the recently introduced \emph{sparse-Lyndon representation} \cite{DBLP:conf/mfcs/Charalampopoulos25} that allows for efficiently grouping the $\tau$-runs for $\tau = \lfloor \frac{1}{3}\log^4 n \rfloor$ according to their Lyndon roots, but representing via their sparse-Lyndon root.
We formalize this discussion in \Cref{prp:group}.

\begin{lemma}[Proposition 36~\cite{DBLP:conf/mfcs/Charalampopoulos25}]
\label{prp:group}
For any string $S$ of length $n$ over an integer alphabet $[0,\sigma)$, with $\sigma = n^{\cO(1)}$, all runs in $S$ can be computed and grouped by equal Lyndon roots in $\cO(n/\log_\sigma n)$ time. For runs with period at most $2\lfloor \frac{1}{18} \log_{\sigma} n \rfloor$, we compute their standard Lyndon representations, while for runs with larger periods we compute their sparse-Lyndon representations.
\end{lemma}

To streamline the subsequent analysis, we hereafter slightly \emph{abuse terminology}: the terms $\tau$-run, Lyndon root, and Lyndon representation will be used uniformly for the two classes of runs considered in this section. In the context of the second class ($\tau = \lfloor \frac{1}{3}\log^4 n \rfloor$), these terms implicitly refer to their sparse-Lyndon counterparts.

This unified terminology is justified by the fact that both representations share analogous properties and can be handled identically within our algorithmic framework.
We group all $\tau$-runs by Lyndon root in $\cO(n/\log_\sigma n)$ time using~\cref{prp:group}; the only difference is that for some groups, we compute the Lyndon representation of runs, while for others, we compute the sparse-Lyndon representation. For further details and a thorough discussion of sparse-Lyndon representations, we refer the reader to~\cite{DBLP:conf/mfcs/Charalampopoulos25}.

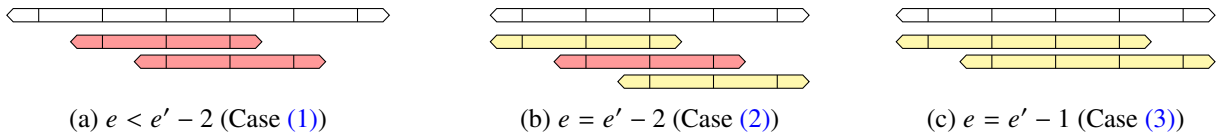
\begin{figure}[htbp]
    \centering
    \begin{subfigure}{0.37\textwidth}
        \centering
        \begin{tikzpicture}[x=24pt, y=-10pt]
            \DrawRun{0}{0}{5}{white};
            \DrawRun{1}{1}{2}{red!40};
            \DrawRun{2}{1.75}{2}{red!40};
            \draw (0,3) circle(0pt);
        \end{tikzpicture}
        \caption{$e < e' - 2$ (Case \ref{lem:per:substring:1})}
    \end{subfigure}
    \begin{subfigure}{0.31\textwidth}
        \centering
        \begin{tikzpicture}[x=24pt, y=-10pt]
            \DrawRun{0}{0}{4}{white};
            \DrawRun{0}{1}{2}{yellow!40};
            \DrawRun{1}{1.75}{2}{red!40};
            \DrawRun{2}{2.5}{2}{yellow!40};
        \end{tikzpicture}
        \caption{$e = e' - 2$ (Case \ref{lem:per:substring:2})}
    \end{subfigure}
    \begin{subfigure}{0.3\textwidth}
        \centering
        \begin{tikzpicture}[x=24pt, y=-10pt]
            \DrawRun{0}{0}{4}{white};
            \DrawRun{0}{1}{3}{yellow!40};
            \DrawRun{1}{1.75}{3}{yellow!40};
            \draw (0,3) circle(0pt);
        \end{tikzpicture}
        \caption{$e = e' - 1$ (Case \ref{lem:per:substring:3})}
    \end{subfigure}
    \caption{Cases \ref{lem:per:substring:1}--\ref{lem:per:substring:3} of \cref{lem:per:substring}, illustrating the potential alignment of a run with exponent $e$ within a run with exponent $e'$. Red alignments are guaranteed to be contained within the longer run; yellow alignments are contained only if the extensions of the runs satisfy specific inequalities.}
    \label{fig:lem:per:substring}
\end{figure}

\begin{lemma}
    \label{lem:per:substring}
    Let $F$ and $F'$ be two runs in $S$ with the same Lyndon root $\lambda$, where $\Lrep(F) = (\lambda,e,\alpha, \beta)$ and  $\Lrep(F') = (\lambda,e',\alpha', \beta')$. We can determine whether $F$ is a \emph{unique substring} of $F'$ as follows:
    
    \begin{enumerate}
        \item If $e < e' -2$, then $F$ is a substring of $F'$, but it is not unique. 
        \label{lem:per:substring:1}

        \item If $e = e' - 2$, then $F$ is a substring of $F'$. It is  unique  if and only if $\alpha >\alpha'$ and $\beta>\beta'$. \label{lem:per:substring:2}
        
        \item If $e = e' - 1$, then $F$ is a unique substring of $F'$ if and only if \emph{either} $\alpha > \alpha'$ \emph{or} $\beta > \beta'$, but not both. (Note: if both $\alpha > \alpha'$ \emph{and} $\beta > \beta'$ hold, $F$ is not a substring of $F'$.) \label{lem:per:substring:3}

        \item If $e = e'$, then $F$ is a unique substring of $F'$ if and only if $\alpha\le\alpha'$ and $\beta\le \beta'$. 
        \label{lem:per:substring:4}
        
        \item If $e > e'$, then $F$ is not a substring of $F'$. \label{lem:per:substring:5}
    \end{enumerate}
\end{lemma}

\begin{proof}
    We analyze each case separately  by aligning the occurrences of the Lyndon roots in $F$ and $F'$ in all possible ways and comparing their extensions; see \cref{fig:lem:per:substring} for illustrative examples. This allows us to determine if  no alignment of $F$ is contained within $F'$ (i.e., $F$ is not a substring of $F'$), exactly one alignment is contained (unique substring), or multiple alignments are contained (not unique).
    
    \begin{enumerate}
        \item If $e < e' - 2$, we can align the first occurrence of the Lyndon root of $F$ with either the second or third occurrence of the Lyndon root of $F'$. In both cases, both extensions of $F$ coincide with complete occurrences of the Lyndon root of $F'$, meaning that $F$ is contained in both alignments and is therefore not a unique substring.
        
        \item If $e = e' - 2$, there are three possible alignments of the Lyndon roots. In the second alignment, both extensions of $F$ coincide with complete occurrences of the Lyndon root of $F'$, so $F$ is a substring of $F'$. In the first alignment, the left extensions of the runs are aligned, meaning that  $F$ is contained in this alignment unless $\alpha>\alpha'$. Similarly, in the third alignment, the same reasoning applies to the right extensions. Therefore, $F$ is a unique substring of $F'$ if and only if $\alpha>\alpha'$ and $\beta>\beta'$.

        \item If $e = e' - 1$, there are two possible alignments. In the first alignment, the left extensions of the runs are aligned while the right extension of $F$ coincides with a complete occurrence of the Lyndon root of $F'$; thus, $F$ is contained in this alignment if $\alpha\le \alpha'$. The same reasoning applies symmetrically for the second alignment. Therefore, $F$ is a unique substring of $F'$ if and only if either $\alpha > \alpha'$ or $\beta > \beta'$, but not both.

        \item If $e = e'$, there is only one possible alignment of the Lyndon roots. For $F$ to be a substring of $F'$, both extensions of $F$ must be shorter than or equal to those of $F'$; that is, $\alpha\le \alpha'$ and $\beta\le \beta'$.

        \item If $e > e'$, any alignment will result in an extension of $F'$ being aligned with a Lyndon root of $F$; therefore, $F$ cannot possibly be a substring of $F'$. \qedhere
    \end{enumerate}
\end{proof}

For the following definitions, let $\mathcal{R}_{\lambda}$ be the set of $\tau$-runs in $S$ sharing the  same Lyndon root $\lambda$. Let $r:=|\lambda|$, and let $e^{\max}_{\lambda}$ denote the maximum exponent among all $\tau$-runs in $\mathcal{R}_{\lambda}$.

\begin{definition}[Mapping $\mathfrak{h}$]
    Let $R \in \mathcal{R}_{\lambda}$ be a run with $\Lrep(R)= (\lambda, e,\alpha, \beta)$. We define the function $\mathfrak{h}$, which maps runs to sets of up to three points as follows:
    \begin{itemize}
        \item If $e < e^{\max}_{\lambda} - 2$, then $\mathfrak{h}(R) = \emptyset$;
        \item If $e = e^{\max}_{\lambda} - 2$, then $\mathfrak{h}(R) = \{(\alpha,\beta)\}$;
        \item If $e = e^{\max}_{\lambda} - 1$, then $\mathfrak{h}(R) = \{(r + \alpha,\beta), (\alpha,r + \beta)\}$;
        \item If $e = e^{\max}_{\lambda}$, then $\mathfrak{h}(R) = \{(2r-1,\beta), (\alpha,2r-1), (r+\alpha,r+\beta)\}$.
    \end{itemize}
\end{definition}

\begin{definition}[Mapping $\mathfrak{g}$]
    Let $p = (x,y)$ be an integer point in $[0,2r - 1]^2$. We define the function $\mathfrak{g}$ as a mapping from such points to candidate runs as follows:
    \begin{itemize}
        \item $\mathfrak{g}(p)$ has Lyndon root $\lambda$;
        \item The left and right extensions of $\mathfrak{g}(p)$ are $\alpha' = x \bmod r$ and $\beta' = y \bmod r$, respectively;
        \item The exponent $e'$ of $\mathfrak{g}(p)$ is determined by the quadrant: $e' =e^{\max}_{\lambda}$, if ($x \ge r$ and $y \ge r$); $e' = e^{\max}_{\lambda} - 1$, if  ($x \ge r$ and $y<r$) or ($x<r$ and $y \ge r$); $e' =e^{\max}_{\lambda} - 2$, if  ($x < r$ and $y < r$).
    \end{itemize}
\end{definition}

\begin{example}
Let $\mathcal{R}_{\lambda}$ be a set of runs sharing the Lyndon root $\lambda = \texttt{ab}$ with period $r=|\lambda|=2$. Assume that $\mathcal{R}_{\lambda}$ consists of three runs, with a maximum exponent of $e^{\max}_{\lambda}=6$:
\begin{itemize}
\item $\Lrep(R_1) = (\texttt{ab}, 6, 1, 1) \implies \mathfrak{h}(R_1) = \{(3, 1), (1, 3), (3, 3)\}$;
\item $\Lrep(R_2) = (\texttt{ab}, 5, 0, 1) \implies \mathfrak{h}(R_2) = \{(2, 1), (0, 3)\}$;
\item $\Lrep(R_3) = (\texttt{ab}, 4, 1, 0) \implies \mathfrak{h}(R_3) = \{(1, 0)\}$.
\end{itemize}
The geometric domain is $[0, 3]^2$, as $2r-1=3$. Applying the mapping function $\mathfrak{h}$  to all runs in  $\mathcal{R}_{\lambda}$ generates the point set $P= \{(3, 1), (1, 3), (3, 3) ,(2, 1), (0, 3), (1, 0)\}$.

Consider the point $(2,2)$. Since $x \ge r$ and $y \geq r$ (specifically, $2 \ge 2$), the mapping $\mathfrak{g}(2,2)$ assigns the exponent $e' = 6$ and extensions $\alpha' = 2 \bmod 2 = 0$ and $\beta' = 2 \bmod 2 = 0$. 
The mapping $\mathfrak{g}$ thus yields the run: $\mathfrak{g}(2, 2) = (\lambda, 6, 0, 0)$. This representation corresponds to the string $\lambda^6 = (\texttt{ab})^6 = \texttt{abababababab}$. Note that the point $(2, 2)$ lies on the primary skyline of~$P$ and is minimal in terms of $x+y$. This geometric property ensures that the corresponding run $\mathfrak{g}(2, 2)$ belongs to exactly one periodic alignment among the runs in $\mathcal{R}_{\lambda}$ (specifically within $R_1$), making $\texttt{abababababab}$ a valid SUS candidate.
\end{example}

\begin{figure}
    \centering
    \begin{subfigure}{0.3\textwidth}
        \centering
        \begin{tikzpicture}[x=2.5cm,y=2.5cm]
            \draw (0.5, -1mm) -- +(0, 2mm);
            \draw (-1mm, 0.5) -- +(2mm, 0);
            \node[anchor=east] at (-1mm, 0.5) {$r$};
            \node[anchor=north] at (0.5, -1mm) {$r$};

            \draw [dotted] (0.5, 1mm) -- (0.5, 1);
            \draw [dotted] (1mm, 0.5) -- (1, 0.5);

            \filldraw[fill=red!8] (0, 0) -- (0.3, 0) -- (0.3, 0.4) -- (0, 0.4);
            \node at (0.3, 0.4) {$\bullet$};
            
            \draw[->] (0, 0) -> (1, 0);
            \draw[->] (0, 0) -> (0, 1);
        \end{tikzpicture}
        
        \caption{$e = e^{\max}_{\lambda} - 2$}
    \end{subfigure}
    \begin{subfigure}{0.3\textwidth}
        \centering
        \begin{tikzpicture}[x=2.5cm,y=2.5cm]
            \node[anchor=east] at (-1mm, 0.5) {$r$};
            \node[anchor=north] at (0.5, -1mm) {$r$};

            \filldraw[fill=red!8] (0, 0) -- (0.8, 0) -- (0.8, 0.4) -- (0.3, 0.4) -- (0.3, 0.9) -- (0, 0.9) -- cycle;
            
            \filldraw[fill=red!22] (0,0) -- (0.3, 0) -- (0.3, 0.4) -- (0, 0.4);

            \draw[->] (0, 0) -> (1, 0);
            \draw[->] (0, 0) -> (0, 1);
            \draw (0.5, -1mm) -- +(0, 2mm);
            \draw (-1mm, 0.5) -- +(2mm, 0);

            \draw [dotted] (0.5, 1mm) -- (0.5, 1);
            \draw [dotted] (1mm, 0.5) -- (1, 0.5);

            \node at (0.8, 0.4) {$\bullet$};
            \node at (0.3, 0.9) {$\bullet$};
        \end{tikzpicture}
        
        \caption{$e = e^{\max}_{\lambda} - 1$}
    \end{subfigure}
    \begin{subfigure}{0.3\textwidth}
        \centering
        \begin{tikzpicture}[x=2.5cm,y=2.5cm]
            \node[anchor=east] at (-1mm, 0.5) {$r$};
            \node[anchor=north] at (0.5, -1mm) {$r$};

            \filldraw[fill=red!8] (0, 0) -- (1.0, 0.0) -- (1.0, 0.4) -- (0.8, 0.4) -- (0.8, 0.9) -- (0.3, 0.9) -- (0.3, 1.0) -- (0.0, 1.0) -- cycle;

            \filldraw[fill=red!22] (0,0) -- (0.8, 0) -- (0.8, 0.4) -- (0.3, 0.4) -- (0.3, 0.9) -- (0, 0.9);

            \draw[->] (0, 0) -> (1, 0);
            \draw[->] (0, 0) -> (0, 1);
            \draw (0.5, -1mm) -- +(0, 2mm);
            \draw (-1mm, 0.5) -- +(2mm, 0);

            \draw [dotted] (0.5, 1mm) -- (0.5, 1);
            \draw [dotted] (1mm, 0.5) -- (1, 0.5);
            
            \node at (1.0, 0.4) {$\bullet$};
            \node at (0.3, 1.0) {$\bullet$};
            \node at (0.8, 0.9) {$\bullet$};
        \end{tikzpicture}
        
        \caption{$e = e^{\max}_{\lambda}$}
    \end{subfigure}
    
    \caption{The mapping $\mathfrak{h}$ from runs to sets of points in $[0, 2r-1]^2$. The light shaded part represents the primary skyline, while dark shaded regions indicate areas dominated by multiple points, where no unique substring can exist.}
    \label{fig:per:h}
\end{figure}
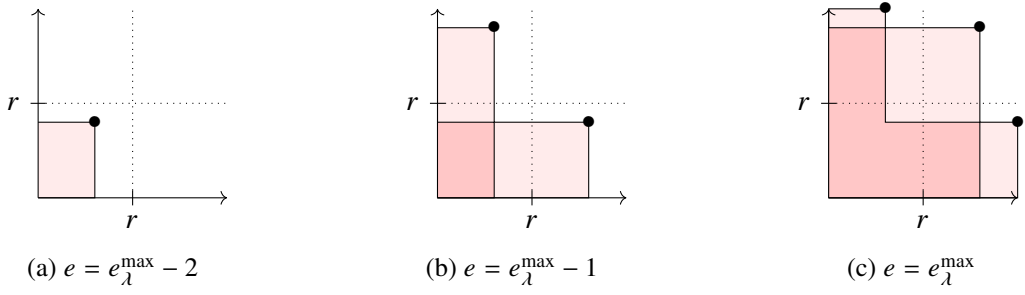

\begin{lemma}
\label{lem:per:skyline}
    Let $R \in \mathcal{R}_{\lambda}$ be a run  with $\Lrep(R)= (\lambda, e,\alpha, \beta)$, 
    and let $p = (x,y)$ be a point in the domain $[0,2r-1]^2$. Then $p$ is in the primary skyline of $\mathfrak{h}(R)$ if and only if $\mathfrak{g}(p)$ is a unique substring of $R$.
\end{lemma}
\begin{proof}
    Let $\mathfrak{g}(p)= R'$, where $\Lrep(R')= (\lambda,e',\alpha',\beta')$. By the definition of $\mathfrak{g}$, we have $\alpha' = x \bmod r$ and $\beta' =y \bmod r$. 
    The quadrant containing $p$ determines the candidate exponent $e'$, while the set of points $\mathfrak{h}(R)$ (and thus its primary skyline) is determined by the run exponent $e$; see \cref{fig:per:h}. We verify the correspondence with the  cases in \cref{lem:per:substring} by distinguishing how $e$ compares to $e'$. 
  
    \begin{itemize}
        \item {\bf Case $e' = e$:}
        The point $p$ can be dominated by exactly one point from $\mathfrak{h}(R)$ with relative coordinates $(\alpha, \beta)$.
        Hence, $p$ belongs to the primary skyline if and only if  it is dominated by this single point, which occurs when $\alpha' \le \alpha$ and $\beta' \le \beta$. This corresponds exactly to the condition for $\mathfrak{g}(p)$ to be a unique substring of $R$ (\cref{lem:per:substring}: Case \ref{lem:per:substring:4}).

        \item {\bf Case $e' = e - 1$:}
        In this case, $p$ can be dominated by at most two points from $\mathfrak{h}(R)$.
        Point $p$ is dominated by the first point if $\alpha' \le \alpha$ and by the second if $\beta' \le \beta$. Thus, $p$ belongs to the primary skyline if it is dominated by exactly one of these points. This requires either $\alpha' \le \alpha$ or $\beta' \le \beta$ but not both. Again, this matches the condition for $\mathfrak{g}(p)$ to be a unique substring of $R$ (\cref{lem:per:substring}: Case~\ref{lem:per:substring:3}).

        \item {\bf Case $e' = e - 2$:} This occurs only  when $e = e^{\max}_{\lambda}$ and $p$ is in the lower-left quadrant ($x,y < r$). Here, $p$ is always dominated by the point $(r+\alpha, r+\beta)$ of $\mathfrak{h}(R)$. For $p$ to be in the primary skyline, it must \emph{not} be dominated by the other two points $(2r-1, \beta)$ and $(\alpha, 2r-1)$.  This lack of dominance occurs if and only if $\alpha' >\alpha$ and $\beta' >\beta$. This happens exactly when $\mathfrak{g}(p)$ is a unique substring of $R$ (\cref{lem:per:substring}: Case~\ref{lem:per:substring:2}). \qedhere
    \end{itemize}
\end{proof}

\begin{lemma}
    \label{lem:per:skyline-reduction}
    Given a set $\mathcal{R}_{\lambda}$ of runs sharing the same Lyndon root $\lambda$, a SUS among these runs can be found by solving a \SkylineProblem instance containing at most $3|\mathcal{R}_{\lambda}|$ points. The corresponding set of (unsorted) points can be constructed in $\cO(|\mathcal{R}_{\lambda}|)$ time. 
\end{lemma}
\begin{proof} 
    Let $e^{\max}_{\lambda}$ denote the maximum exponent among all runs in $\mathcal{R}_{\lambda}$.
    By \cref{lem:per:substring}, we know that any SUS within these runs must have an exponent equal to $e^{\max}_{\lambda} - 2$, $e^{\max}_{\lambda} - 1$ or~$e^{\max}_{\lambda}$.
    Such a substring thus corresponds to $\mathfrak{g}(p)$ for some point $p \in [0,2r - 1]^2$, with $r=|\lambda|$.

    By \cref{lem:per:skyline}, $p$ must lie on the primary skyline of $\mathfrak{h}(R)$ for some $R \in \mathcal{R}_{\lambda}$. Moreover, $p$ cannot belong to the primary skylines of multiple runs; otherwise, $\mathfrak{g}(p)$ would not be unique overall. Hence, $p$ must be dominated by exactly one point among all sets $\mathfrak{h}(R)$, which can be identified  by solving the \SkylineProblem instance on the multiset $\bigsqcup_{R\in \mathcal{R}_{\lambda}} \mathfrak{h}(R)$.
    Since $|\mathfrak{h}(R)| \le 3$ for every $R$, this instance has size $\cO(|\mathcal{R}_{\lambda}|)$. Finally, if the resulting point from this instance is $(x,y)$, it corresponds to a substring of length $r\cdot(e^{\max}_{\lambda}-2)+x+y$.
\end{proof}

\begin{lemma}
\label{lem:per}
Given an instance of \SUS and an integer $\tau$, such that $\tau = \lfloor\frac{1}{15}\log_\sigma n \rfloor$ or $\tau = \lfloor\frac13\log^4 n \rfloor$, we can compute a SUS of $S$ in $\cO(n/\log_\sigma n)$ time 
if it has length $\ell \geq 3\tau$ and period $p \le \frac13 \tau$.
\end{lemma}

\begin{proof}
    By \cref{lem:per:find-roots} and \cref{prp:group}, all runs in $S$ can be computed and grouped by their Lyndon roots in $\cO(n/\log_\sigma n)$ time.
    Each $\tau$-run is encoded in $\cO(1)$ space. 
    Next, using \cref{lem:per:skyline-reduction}, we construct \SkylineProblem instances for all groups. Since each $\tau$-run generates a constant number of points, the total number of points across all instances is $\cO(n/\tau)$. These points are computed in $\cO(n/\tau)$ time as per \cref{lem:per:find-roots}. 
    To solve all \SkylineProblem instances in linear time (\cref{lem:skyline-alg}), the points for each $\mathcal{R}_{\lambda}$ instance must be sorted along one axis. This is achieved by globally sorting  all generated points using bucket sort. 
    Because the coordinates $x$ and $y$ for any point in an instance for root $\lambda$ are bounded by $2|\lambda|-1 < 2\tau$, the bucket sort can be performed in $\cO(n/\tau + \tau)$ time. After sorting,  the points are regrouped per their Lyndon roots. 
    For either chosen value of $\tau$, we have  $\cO(n/\tau + \tau)=\cO(n/\log_\sigma n)$. This yields a total running time of $\cO(n/\log_\sigma n)$.
\end{proof}

\subsection{Putting It All Together}

\sus*
\begin{proof}
    We have described algorithms to find SUSs with differing lengths and periods: \cref{lem:short} handles any length up to $\frac15 \log_\sigma n$ with any period; \cref{lem:med} handles lengths in the range $[\frac15 \log_\sigma n, 2^{\sqrt{\log n}}]$ with period greater than $\frac{1}{45}\log_\sigma n$; and \cref{lem:long} handles lengths above $\log^4n$ with period greater than $\frac19 \log^4n$.
    Finally, with parameter values $\tau:=\lfloor\frac{1}{15}\log_\sigma n \rfloor$ and $\tau:= \lfloor\frac13 \log^4n \rfloor$, \cref{lem:per} handles the larger length ranges with smaller periods. We run each of the algorithms and return the globally shortest substring output across all instances. The total running time is $\cO({n\log \sigma}/{\sqrt{\log n}})$ which is asymptotically dominated by the complexity of the medium-length aperiodic case (\cref{lem:med}).
\end{proof}

\section{Reduction of Packed SUS to Binary Alphabets}\label{sec:bin}
In this section, we show that the \SUS problem over an integer alphabet can be reduced to the binary case in optimal time without any blow-up in the space needed to represent the obtained string.
This reduction allows us to assume that $\sigma=2$ for \SUS algorithms in the packed setting.

\begin{lemma}
\label{lem:binary}
Any instance of \SUS for a string $S$ of length $n$ over alphabet $[0,\sigma)$ can be reduced to an instance of \SUS for a string $S'$ of length $\cO(n\log\sigma)$ over a binary alphabet. If both strings are packed, then the reduction works in $\cO\left(\frac{n\log \sigma}{\log n}\right)= \cO\left(\frac{n}{\log_{\sigma} n}\right) $ time.
\end{lemma}
\begin{proof}
 Let $k:=\lceil\log\sigma\rceil+2$, and let $b(a)$ for $a\in[0,\sigma)$ be the binary representation of letter~$a$ padded to length $k-2$. We define the morphism  $\pi(a)=0^k \cdot 1\cdot b(a)\cdot 1$.
We construct $S'$ by concatenating $\pi(S[0])\pi(S[1])\cdots\pi(S[n-1])$ and appending $0^k$.

   \proofsubparagraph{Correctness.} The use of the $0^k$ blocks ensures that, for any string $T$, $\pi(T)$ can occur only at positions equivalent to $0\bmod{2k}$ in $S'$. Consequently, a string $T$ is unique in $S$ if and only if $\pi(T)$ is unique in $S'$. We now show that a SUS of $S'$ necessarily corresponds to a SUS of~$S$ (and lets us retrieve it).
   
    Let $U$ be a SUS of $S'$.
    If $U$ is a substring of some fragment $\pi(a)0^k$ then its uniqueness implies that the letter $a$ is unique in $S$, and thus that $a$ is a SUS of $S$.
    Otherwise, $U$ must contain at least one occurrence of~$0^k$.
    Additionally, $U$ cannot start (symmetrically end) inside a fragment equal to $0^k$ as removing the first $0$ would yield a shorter unique substring (this letter is fixed by the position modulo $2k$), 
    contradicting the minimality of $U$. 
    Thus, $U$ must be of the form $x 0^k U' 0^k y$ (or $x 0^k y$) with $|x|,|y| \in [1, k]$.
    $U$ is a substring of $\pi(T)$, for some unique substring $T$ of $S$ of length $\ell$ with $|0^k U'| = 2k (\ell-2)$---we claim that this $T$ is a shortest unique substring of $S$.
   
   To prove this, let $T'$ be a unique substring of $S$. We easily obtain a unique substring $U'=\pi(T')[k\dd 2k |T'|)$ of $S'$ (by removing the prefix $0^k$ from $\pi(T')$).
   By the minimality of $U$, we have $2k |T'|-k=|U'|\ge |U|= 2k(\ell-2) +k +|x|+|y|\ge 2k(\ell-2)+k+2=2k (\ell-1)-k+2$.
   Therefore, $|T'|\ge\ell$ and hence, $T$ is a SUS of $S$.

   \proofsubparagraph{Complexity.}
    We have $|S'|=2k |S|+k= (2n+1) k = (2n+1) (\lceil\log\sigma\rceil +2)=\cO(n\log\sigma)$.
   As the reduction proceeds by constructing length-$2k$ blocks letter by letter, it can be trivially applied to each machine word in $\cO(1)$ time (we insert $1\cdot0^k\cdot 1$ every $k-2$ bits).
   \end{proof}

\begin{example}
    Let $S=\texttt{gctctca}$ with $\sigma=4$, $b(\texttt{a})=\texttt{00}$, $b(\texttt{c})=\texttt{01}$, $b(\texttt{g})=\texttt{10}$, and $b(\texttt{t})=\texttt{11}$. We have $k=\lceil\log\sigma\rceil+2=4$,
    $\pi(\texttt{a})=\texttt{00001001}$, $\pi(\texttt{c})=\texttt{00001011}$, $\pi(\texttt{g})=\texttt{00001101}$, and $\pi(\texttt{t})=\texttt{00001111}$. We construct the new instance $S'$ of length $2k\cdot |S|+k=60$ as follows:
    \begin{multline*}
S' = \pi(\texttt{g})\pi(\texttt{c})\pi(\texttt{t})\pi(\texttt{c})\pi(\texttt{t})\pi(\texttt{c})\pi(\texttt{a}) \cdot \texttt{0000} \\
= \texttt{000011010000101100001111000010110000111100001011000010010000}.
\end{multline*}

The SUSs of $S'$ are of length $4$: \texttt{1101} and \texttt{1001} correspond to \texttt{g} and \texttt{a}, respectively.
\end{example}

\section{Computing a Shortest Absent Substring}\label{sec:SAS}

We begin by defining an auxiliary problem closely related to \SUS.

\defproblem{\SUSDiff}{Two strings $S_1$ and $S_2$ with $n = |S_1| + |S_2|$ over an integer alphabet $\Sigma=[0,\sigma)$.}{A shortest substring of $S_1$ that does not occur in $S_2$ (if one exists).}

Similar to the solution for \SUS, we solve the \SUSDiff problem by decomposing it into four cases based on the length $\ell$ and the period $p$ of the sought substring of $S_1$, obtaining the following result.

\begin{theorem}\label{thm:difference}
Any instance of \SUSDiff can be solved in $\cO(n\log \sigma/ \sqrt{\log n})$ time when $S_1$ and $S_2$ are given in packed representation.
\end{theorem}

We explain how one can adapt our algorithm for the \SUS problem to obtain \Cref{thm:difference} in \Cref{sec:exclusive}.
We next present an efficient construction of de Bruijn sequences in the packed setting, which may be of independent interest.

\begin{definition}[De Bruijn Sequence~\cite{deBruijn1946}]
A \emph{de Bruijn sequence} of order $k$ over an alphabet~$\Sigma$ of size $\sigma$ is a string of length $\sigma^k+k-1$ in which every string from $\Sigma^k$ occurs exactly once.
\end{definition}

\begin{lemma}\label{lem:deBruijn}
A packed de Bruijn sequence of order $k$ over the integer alphabet $[0,\sigma)$ for $\sigma\ge 2$ can be constructed in $\cO(\frac{\sigma^k}{k})=\cO(\frac{n}{\log_{\sigma} n})$ time, where $n$ is the length of the sequence. In particular, we can construct its prefix of length $\ell$ in $\cO(\frac{\ell}{\log_{\sigma} n} + 1)$ time. 
\end{lemma}

\begin{proof}
A \emph{Lyndon word} is a string that is lexicographically strictly smaller than all of its proper suffixes.
As noted by Fredricksen and Maiorana~\cite{DBLP:journals/dm/FredricksenM78}, the concatenation of all Lyndon words whose length divides $k$, listed in lexicographical order, forms a de Bruijn sequence.
Duval~\cite{DBLP:journals/tcs/Duval88} provided an algorithm to generate all Lyndon words of length at most $k$ in lexicographical order, generating each word by modifying the previous one.

We achieve the stated running time in the packed setting, by observing that $k \leq \lceil\log_\sigma n\rceil$, and using $\cO(1)$ machine words to iterate over all Lyndon words in $[0,\sigma)^{\le k}$.

In what follows, we describe Duval's algorithm~\cite{DBLP:journals/tcs/Duval88} without proving its correctness; we only explain how it can be performed efficiently in the packed setting.
We will return a string $S$, initialized as~$0$.
We then maintain a length-$k$ string $w$ initialized as $0^k$ and a binary string $w'$ whose $i$th bit is set if and only if $w[i] = \sigma - 1$ (we update $w'$ together with $w$).
We repeatedly apply the following steps of the generation loop while $w[0] \neq \sigma - 1$:
\begin{enumerate}
\item Locate $j=\max\{i: w[i]\neq \sigma - 1\}$. This is done by finding the rightmost $0$ in $w'$ in $\cO(1)$ time using standard bitwise operations.
\item Replace $w[j]$ with $w[j]+1$ in $\cO(1)$ time.
\item The prefix $w[0 \dd j]$ is now a Lyndon word. If $(j+1)$ divides $k$, we append $w[0\dd j]$ to $S$. This takes $\cO(1)$ time as the appended string fits into $\cO(1)$ machine words.
\item Replace $w$ with the length-$k$ prefix of $(w[0\dd j])^\infty$.
This can be naively performed using $\cO(k/(j+1))$ operations.
Consider a potential function $\pi(w)=|\{i:w[i] = \sigma - 1\}|$, noting that Step~$2$ increases $\pi(w)$ by at most $1$. Conversely, in Step~$4$, values $w[i]=\sigma - 1$ for $i>j$ are replaced with $w[i \bmod (j+1)]$. Since $w[0] < \sigma - 1$ (the algorithm terminates when $w[0]$ reaches $\sigma - 1$), $\Omega(k/(j+1))$ copies of $\sigma-1$ are overwritten by setting $w = (w[0\dd j])^\infty$.
As $\pi(w)\ge 0$ at all times, the total number of operations we perform in Step 4 is asymptotically upper bounded by the number of times Step 2 is executed, and hence the amortized running time of Step 4 is $\cO(1)$.
\end{enumerate}

After the last iteration, we complete the sequence by appending $0^{k-1}$.
Observe that, for each integer $\ell$ there are at most $\frac{\sigma^\ell}{\ell}$ Lyndon words of length $\ell$, and hence at most $\sum_{\ell\le k}\frac{\sigma^\ell}{\ell}\le 3\frac{\sigma^k}{k}$ Lyndon words of length at most $k$. This proves that there are $\cO(\frac{\sigma^k}{k})$ iterations of the loop, and the algorithm takes $\cO(\frac{n}{\log_{\sigma} n})$ time (where $n=\sigma^k+k-1$ and $k=\lfloor\log_{\sigma} n\rfloor$) as claimed.
To output a length-$\ell$ prefix of this sequence we terminate the algorithm once the length of the output sequence reaches $\ell$, potentially removing up to $k-1$ excessive letters from the final Lyndon word.  Since $k/\log_\sigma n = \cO(1)$, the total time complexity is
\[\cO\left(\frac{\ell+k}{\log_{\sigma} n} + 1 \right)=\cO\left(\frac{\ell}{\log_{\sigma} n}+1\right).\]\qedhere 
\end{proof}

\sas* 

\begin{proof}
Let $k:=\lfloor\log_\sigma n\rfloor+1$. Since $\sigma^k >n$, the length of any SAS is at most $k$.
We first check for a SAS of length at most $k-1$ by applying \cref{thm:difference} to a packed de Bruijn sequence $S_1$ of order $k-1$ (\cref{lem:deBruijn}) and $S_2:=S$.
If this yields a substring of length at most $k - 1$, we are done. Otherwise, if $\sigma^k = \cO(n)$, we create a full packed de Bruijn sequence of order $k$ and repeat the process to find a SAS of length $k$.
If $\sigma^k = \omega(n)$, we instead
generate a prefix $S'_1$ of a de Bruijn sequence of order $k$ of length $n+1$ in $\cO(\frac{n}{\log_\sigma n})$ time. $S'_1$ contains $n-k+2$ distinct substrings of length $k$. Since $S_2$ has at most $n-k+1$ such substrings, at least one substring of $S'_1$ must be absent from $S_2$. We find this witness of length $k$ using \cref{thm:difference} for $S'_1$ and $S_2$.
The calls to \cref{thm:difference} dominate the total running time, which is $\cO\left(\frac{n \log \sigma}{\sqrt{\log n}}\right)$.
\end{proof}

\subsection{Computing a Shortest Exclusive Substring}\label{sec:exclusive}

Let us fix an instance of the \SUSDiff problem.
We first decide if there exists a substring of $S_1$ that does not occur in $S_2$ in $\cO(n/\log_{\sigma} n)$ time by checking if $S_1$ occurs in $S_2$~\cite{DBLP:journals/tcs/Ben-KikiBBGGW14}. If it does, then $S_2$ contains all the substrings of $S_1$; otherwise, the full~$S_1$ is an exclusive substring.
We henceforth assume that an exclusive substring exists.

The algorithm presented below is an adaptation of the one presented in \cref{sec:SUS} for the \SUS problem.
We first define an appropriate variant of the \SkylineProblem problem and present an algorithm for it.

\defproblem{\SkylineDiff}{Two sets of points $P_1,P_2$ in  $\mathbb{Z}^2_{\ge 0}$.}{A point $(x, y) \in \mathbb{Z}^2_{\ge 0}$ (if one exists) that is dominated by at least one point $p_1 \in P_1$, is not dominated by any point $p_2 \in P_2$, and minimizes $x+y$.}

\begin{lemma}
    \label{lem:sas:skyline}
    Any instance of \SkylineDiff can be solved in $\cO(|P_1| + |P_2|)$ time, provided that $P_1$ and $P_2$ are given as lists sorted with respect to one of the two coordinates.
\end{lemma}
\begin{proof}
   We adapt the right-to-left scanning algorithm from \cref{lem:skyline-alg}. In the original problem, the goal was to find a point dominated by exactly one point in a set $P$, which required tracking the highest and second-highest $y$-values ($y_{\max}$ and $y'$). For \SkylineDiff, the requirement is to find a point dominated by \emph{at least one} point in $P_1$ and \emph{no} points in $P_2$.
   
   Assume $P_1$ and $P_2$ are sorted by $x$-coordinate in non-increasing order. We perform a joint scan from right to left, maintaining $y^{(1)}_{\max}$ and $y^{(2)}_{\max}$—the maximum $y$-coordinates encountered so far in $P_1$ and $P_2$, respectively. Both are initialized to $-1$. For each unique $x$-coordinate $x_i$ in $P_1 \cup P_2$, we evaluate the candidate $x^\star = x_i + 1$ \emph{before} updating the maxima with the points at $x_i$. For this $x^\star$, a valid $y^\star$ must satisfy two conditions:
    \begin{enumerate}
        \item $y^\star \le y^{(1)}_{\max}$ (the point is dominated by $P_1$);
        \item $y^\star > y^{(2)}_{\max}$ (the point is not dominated by $P_2$).
    \end{enumerate}
    Such a $y^\star$ exists if and only if $y^{(1)}_{\max} > y^{(2)}_{\max}$. To minimize the sum $x+y$ for the current $x^\star$, we choose the smallest possible $y$-coordinate: $y^\star = \max\{0, y^{(2)}_{\max} + 1\}$. We then update $y^{(1)}_{\max}$ and $y^{(2)}_{\max}$ with all points $(x_i, y) \in P_1 \cup P_2$ and continue the scan. After processing all points, we perform a final check for $x^\star = 0$.
    
    Since the scan visits each point in $P_1 \cup P_2$ once and updates the maxima in $\cO(1)$ time, the total running time is $\cO(|P_1| + |P_2|)$.
    \end{proof}

The algorithm underlying \Cref{thm:difference} decomposes the \SUSDiff  problem into the same four cases---short, medium aperiodic, long aperiodic, and periodic---as our algorithm for the \SUS problem.

Below, we sketch how the algorithm for each of the cases of the \SUS problem can be adapted for the corresponding case of the \SUSDiff problem.

\subsubsection{Short Case}
Let $\ell := \lfloor\frac15 \log_\sigma n \rfloor$. 
As in \cref{lem:short}, we take every length-$2\ell$ fragment of $S_1$ and $S_2$, starting at positions equivalent to $0$ modulo $\ell$. 
For each distinct string encountered as a fragment, we store at most one representative from $S_1$ and $S_2$.
Each distinct fragment $F_i$ is labeled with $1$ if it occurs in $S_2$ and with $0$ if it occurs \emph{exclusively} in~$S_1$. Using bucket sort, we produce a lexicographically sorted list $F_1 < F_2 < \dots < F_q$ of distinct fragments, each inheriting its label.

Next, we construct the string $S' = F_1\#_1F_2\#_2\dots \#_{q-1}F_q\#_q$; where each non-$\#$ position inherits the label of its fragment.
Our goal is to find the shortest substring $P$ of $S'$, not containing any~$\#$, such that all occurrences of $P$ of $S'$ correspond to positions labeled $0$. This ensures that $P$ appears in $S_1$ but not in $S_2$ within the sampled positions.
To achieve this, we build the suffix tree $\ST(S')$ and propagate the labels from the leaves to internal nodes: a leaf is labeled $0$ or $1$ based on its starting fragment, and an internal node is labeled~$1$ if any of its children has label~$1$, and $0$ otherwise.
We prune all branches starting with a $\#$, obtaining a trie of the fragments $F_1,\dots,F_q$, where each prefix is labeled based on whether it occurs solely in $S_1$ or not. 

A shortest exclusive substring of length at most $\ell$, if one exists, is the path label of node~$u$ with label $0$ that minimizes the value $\sd(\textsf{parent}(u))+1$ (which is at most $\ell$). This string is the prefix of the path label of $u$ that is exactly one letter longer than the path label of its $1$-labeled parent. Note that if no node has label $0$, then there is no exclusive substring of length at most $\ell$.

The entire process runs in  $\cO(n^{2/5} + n/\log_\sigma n)=\cO(n/\log_\sigma n)$ time, similarly to \cref{lem:short}: we consider $\cO(n/\log_\sigma n)$ fragments, yielding $\cO(n^{2/5})$ distinct strings, which are then processed in $\cO(n^{2/5} \log_\sigma n)$ total time.

\subsubsection{Medium Aperiodic Case}
We first define the following analogous problem of \PairSUS:  

\defproblem{\PairSES}{%
Compacted tries $\mathcal{T}(\mathcal{F}_1)$ and $\mathcal{T}(\mathcal{F}_2)$  of $\mathcal{F}_1, \mathcal{F}_2 \subseteq \Sigma^*$,
and two (multi)sets $\mathcal{S}_1, \mathcal{S}_2 \subseteq \mathcal{F}_1\times \mathcal{F}_2$ 
with $|\mathcal{S}_1|, |\mathcal{S}_2|, |\mathcal{F}_1|, |\mathcal{F}_2| \le N$.}{Integers $\ell_1,\ell_2\ge0$ and a pair $(P_1,P_2) \in \mathcal{S}_1$, such that $(\ell_1 + \ell_2)$ is minimized and for all  $(Q_1, Q_2) \in \mathcal{S}_2$, either $\LCP(P_1,Q_1) < \ell_1$ or $\LCP(P_2, Q_2) < \ell_2$.}

In other words, we want to find a shortest pair of strings that are prefixes of some pair in~$\mathcal{S}_1$ but not of any pair in $\mathcal{S}_2$. For both prefix families and $(\alpha,\beta)$-families, we show how to adapt the algorithms from \cref{sec:medium} to \PairSES.

\subparagraph{Solution for prefix families.} We consider the case in which $\mathcal{S}_1 = \{(U_1,V_1),\ldots,(U_N,V_N)\}$ and $\mathcal{S}_2 = \{(U'_1,V'_1),\ldots,(U'_{N'},V'_{N'})\}$ and $\mathcal{S}_1 \cup \mathcal{S}_2$ is a prefix family, i.e. all $U_i$ and $U'_i$ are prefixes of some common string and can thus be represented by their lengths. We first present the analogue of \cref{lem:med:prefix-formula}:

\begin{lemma}
    Consider an instance of \PairSES in which $\mathcal{S}_1$ and $\mathcal{S}_2$ are both prefix families. For any integer $\ell_1 > 0$, let $I_{\ell_1} = \{i \mid |U_i| \ge \ell_1\}$ and $I'_{\ell_1} = \{i \mid |U'_i| \ge \ell_1\}$. Further let
    \begin{equation*}
        f(\ell_1) = \min_{i \in I_{\ell_1}} \left( \max_{j \in I'_{\ell_1}} \LCP(V_i, V'_j) \right),
    \end{equation*}
    \noindent with $f(\ell_1)=0$ if $I'_{\ell_1} = \emptyset$. The pair $(\ell_1, f(\ell_1) + 1)$ minimizing the sum $\ell_1 + f(\ell_1) + 1$ is an optimal solution.
\end{lemma}
\begin{proof}
    The proof borrows heavily from \cref{lem:med:prefix-formula}.

    \proofsubparagraph{Feasibility.} For a fixed $\ell_1$, let $i^* \in I_{\ell_1}$ be the index attaining the minimum in $f(\ell_1)$ and let $\ell_2 = f(\ell_1) + 1$. We will show that $(U_{i^*}, V_{i^*})$ with lengths $(\ell_1,\ell_2)$ is a feasible solution for \PairSES. For this it must be that every $(U'_j,V'_j) \in \mathcal{S}_2$ has $\LCP(U_{i^*},U'_j) < \ell_1$ or $\LCP(V_{i^*}, V'_j) < \ell_2$. There are two cases for each index $j$:
    \begin{itemize}
        \item Case 1: $j \notin I'_{\ell_1}$. Then by definition $|U'_j| < \ell_1$, fulfilling the first condition.
        \item Case 2: $j \in I'_{\ell_1}$. Then by  definition of $f(\ell_1)$, we have $\LCP(V_{i^*},V'_j)\le \max_{j' \in I'_{\ell_1}} \LCP(V_{i^*}, V'_{j'}) $ $\le \ell_2 - 1$, fulfilling the second condition.
    \end{itemize}

    \proofsubparagraph{Optimality.} Let $(\ell_1^*, \ell_2^*)$ be an optimal solution with witness $(U_i,V_i)$. Since this solution must be feasible, we have for every $j \in I'_{\ell^*_1}$ that $\LCP(V_i, V'_j) \le \ell^*_2 - 1$. Thus we have

    \begin{equation*}
        \max_{j \in I'_{\ell^*_1}} \LCP(V_i, V'_j) \le \ell^*_2 - 1.
    \end{equation*}

    From the definition of $f(\ell^*_1)$ it follows that $\ell^*_2 \ge f(\ell^*_1)+1$, which implies that $\ell^*_1+\ell^*_2 \ge \ell^*_1 + f(\ell^*_1) + 1$. 
\end{proof}

We define $\maxLCP[i] := \max_{j,|U'_j|\ge|U_i|}\LCP(V_i,V'_j)$ for $(U_i,V_i) \in \mathcal{S}_1$ and adapt \cref{lem:med:prefix-max-lcp} to compute these values in $\cO(N + N')$ time using auxiliary arrays $\prev$ and $\nextt$ on a sorted list of pairs from both families.

A proof for the following lemma is nearly identical to that of \cref{lem:med:pairsus-prefix}.
\begin{lemma}
\label{lem:ses-med-prefix}
    An instance of $\PairSES$ in which $\mathcal{S}_1 \cup \mathcal{S}_2$ is a prefix family can be solved in $\cO(N + N')$ time.
\end{lemma}

\subparagraph{Solution for $(\alpha,\beta)$-families.}

\begin{lemma}
\label{lem:sas:med:lcps}
    Let $\mathcal{S}$ be a lexicographically sorted list of $N$ strings, and let $L$ be an array of integers, where $L[i] = \LCP(\mathcal{S}[i-1],\mathcal{S}[i])$, for all $i \in [1,N)$. Let $f : [0,N) \to \{0,1\}$ be a coloring function on the strings of $\mathcal{S}$, with $f(i) = 0$ for at least one $i$. In $\cO(N)$ time, we can find the shortest string $P$ that is a prefix of some $\mathcal{S}[i]$ with $f(i) = 0$, such that $P$ is not a prefix of any $\mathcal{S}[j]$ with $f(j) = 1$.
\end{lemma}
\begin{proof}
    For each $i \in [0,N)$ with $f(i)=0$,  let $L_p(i) := \LCP(\mathcal{S}[j],\mathcal{S}[i])$ where $j$ is the largest index $j < i$ such that $f(j)=1$. Similarly, let $L_n(i) := \LCP(\mathcal{S}[k],\mathcal{S}[i])$ where $k$ is the smallest index $k > i$ such that $f(k)=1$. If no such $j$ or $k$ exists, the corresponding value is $0$. All $L_p$ and $L_n$ values can be computed in linear time: if $f(i-1)=0$, then $L_p(i)=\min\{L_p(i-1),L[i]\}$ by \cref{lem:min-lcp}; if $f(i-1)=1$, then $L_p(i) = L[i]$. The  computation of $L_n$ is analogous, but in the reverse direction.

    By \cref{lem:min-lcp}, for any $i,j$ with $f(i)=0$ and $f(j)=1$, we have $\LCP(\mathcal{S}[i],\mathcal{S}[j]) \le \max\{L_p(i),L_n(i)\}$. Hence, any prefix of $\mathcal{S}[i]$ longer than this value  cannot be shared with any $\mathcal{S}[j]$ where $f(j)=1$. Thus, the length of the shortest desired prefix is 
    \[\min_{i \in [0,N),f(i)=0}(\max\{L_p(i),L_n(i)\}+1),\] which can be computed in $\cO(N)$ time.
\end{proof}

\begin{example}
Let $\mathcal S = \{\texttt{ape}, \texttt{apple}, \texttt{bacon}, \texttt{band}, \texttt{bank}\}$ be a set of $N=5$ lexicographically sorted strings, and let the coloring function be
$f(0)=0$, $f(1)=1$, $f(2)=0$, $f(3)=0$, and $f(4)=1$.
Our goal is to find the shortest prefix of a string with color $f=0$ that is not a prefix of any string with color $f=1$.

\medskip

\begin{center}
\begin{tabular}{c|l|c|c|c|c|c}
$i$ & $S[i]$ & $f(i)$ & $L[i]$ & $L_p(i)$ & $L_n(i)$ & $1+\max\{L_p(i),L_n(i)\}$ \\
\hline
0 & \texttt{ape}   & 0 & 0 & 0 & 2 & 3 \\
1 & \texttt{apple} & 1 & 2 & -- & -- & -- \\
2 & \texttt{bacon} & 0 & 0 & 0 & 2 & 3 \\
3 & \texttt{band}  & 0 & 2 & 0 & 3 & 4 \\
4 & \texttt{bank}  & 1 & 3 & -- & -- & -- \\
\end{tabular}
\end{center}

\medskip
\raggedright
Among the indices with $f(i)=0$, the minimum length is $\min\{3,3,4\}=3$. This corresponds to the length-$3$ prefixes of $\mathcal S[0]$ and $\mathcal S[2]$, namely \texttt{ape} and \texttt{bac}. These strings appear in the $f=0$ set but do not occur as prefixes in the $f=1$ set.
\end{example}

\subparagraph{Wrapping up.}
We construct $\tau$-synchronizing sets for $S_1$ and $S_2$, with $\tau := \lfloor\frac{1}{15} \log_\sigma n \rfloor$.
We use these synchronizing sets to identify $\tau$-runs in both strings and group these by their suffixes, in order to construct \PairSES instances akin to \cref{lem:med}. Each instance is a prefix family and can be solved using \cref{lem:ses-med-prefix}. For candidates that do not start in a $\tau$-run, we take prefix-suffix pairs according to the synchronizing sets and apply \cref{lem:med:wavelet} to their union to build a wavelet tree with suffix lists at each node. Each node  additionally stores a bit vector indicating  whether each attached suffix belongs to $S_1$ or~$S_2$. This bit vector acts as the coloring function $f$ in \cref{lem:sas:med:lcps}. These can be computed without increasing the asymptotic time and space complexities, as in \cref{lem:med:lcs}. Finally, we apply \cref{lem:sas:med:lcps} to each node of the wavelet tree to obtain, for every prefix, the shortest corresponding suffix occurring in $S_1$ but not in $S_2$.

\subsubsection{Long Aperiodic Case}
We apply the same technique as in \cref{lem:long}. First, we construct $\tau$-synchronizing sets for $S_1$ and $S_2$, with $\tau:= \lfloor\frac13 \log^4 n \rfloor$. Following the original algorithm, we build two compacted tries~$T$ and $T^R$ containing the suffixes and prefixes of  strings starting or ending at anchor positions. The tries include strings from both $S_1$ and $S_2$; we distinguish them by coloring their leaves based on their source string. Since $n=|S_1|+|S_2|$, these tries occupy $\cO(n\log \sigma / \log^4 n)$ space. The rest of the algorithm proceeds identically to the \SUS case, except when mapping pairs of heavy paths to the 2D domain. We use \SkylineDiff to identify points dominated by $S_1$-leaves that are not dominated by any $S_2$-leaves. This ensures that the resulting substrings are exclusive to $S_1$.

\subsubsection{Periodic Case}

We apply the same technique as in \cref{lem:per}. First, we group all $\tau$-runs from both strings by their Lyndon root. For each distinct root, we have a set of runs from $S_1$ and a set from~$S_2$.
For each Lyndon root, we construct a \SkylineDiff instance where $P_1$ represents the periods/length available in $S_1$ and $P_2$ represents those in $S_2$.  Solving these instances allows us to determine the shortest periodic string that appears in $S_1$ but not in $S_2$.
Since our solution for the \SkylineDiff  (\cref{lem:sas:skyline}) matches the  asymptotic complexity of \SkylineProblem  (\cref{lem:skyline-alg}), the overall time complexity remains identical to that of \cref{lem:per}.

\bibliographystyle{plainurl}
\bibliography{references}

\end{document}